\documentclass[a4paper,american,numberwithinsect]{lipics-v2018-modified}

\usepackage{microtype}

\usepackage{xcolor}
\usepackage{tikz}
\usetikzlibrary{arrows}
\usepackage{hyperref}
\usepackage{ifxetex}
\usepackage{paralist}

\usepackage{xspace}
\usepackage{amssymb}
\usepackage{mathtools}

\newcommand\etal{~et~al.\xspace}
\newcommand\opcit{~op.~cit.\xspace}
\newcommand\theaplayer{Agnetha\xspace}
\newcommand\theeplayer{Elvis\xspace}

\newcommand\setcomp[2]{\left\{{#1}\ \left|\ {#2}\right.\right\}}
\newcommand\set[1]{\left\{{#1}\right\}}
\newcommand\powerset[1]{\mathbb{P}(#1)}
\newcommand\tup[1]{\brac{#1}}
\newcommand\brac[1]{({#1})}
\newcommand\ap[2]{{#1}\mathord{\brac{#2}}}
\newcommand\idxi{i}
\newcommand\idxj{j}

\newcommand\numof{\ell}

\newcommand\naturals{\mathbb{N}}
\newcommand\sizeof[1]{\left|{#1}\right|}

\DeclarePairedDelimiter{\ceil}{\lceil}{\rceil}
\newcommand\nexp[1]{\text{Exp}_{#1}}

\newcommand\cpds{\mathcal{C}}
\newcommand\orcpds{\orof{\cpds}}
\newcommand\crcpds{\crof{\cpds}}
\newcommand\prcpds{\prof{\cpds}}
\newcommand\controls{\mathcal{P}}
\newcommand\orcontrols{\mathcal{Q}}
\newcommand\crcontrols{\mathcal{Q}}
\newcommand\prcontrols{\mathcal{Q}}
\newcommand\controlset{P}
\newcommand\cpdsrules{\mathcal{R}}
\newcommand\orcpdsrules{\orof{\mathcal{R}}}
\newcommand\crcpdsrules{\crof{\mathcal{R}}}
\newcommand\prcpdsrules{\prof{\mathcal{R}}}
\newcommand\oralphabet{\orof{\alphabet}}
\newcommand\cralphabet{\crof{\alphabet}}
\newcommand\pralphabet{\prof{\alphabet}}
\newcommand\cpdsord{n}
\newcommand\opord{k}
\newcommand\cops[1]{\mathcal{O}_{#1}}

\newcommand\cha{a}
\newcommand\chb{b}

\newcommand\rew[1]{\mathit{rew}_{#1}}
\newcommand\cpush[1]{\mathit{push}_{#1}}
\newcommand\push[1]{\mathit{push}_{#1}}
\newcommand\pop[1]{\mathit{pop}_{#1}}
\newcommand\collapse{\mathit{collapse}}
\newcommand\noop{\mathit{nop}}
\newcommand\dontcare{*}

\newcommand\cpdsrule[4]{\tup{{#1},{#2},{#3},{#4}}}

\newcommand\control{p}
\newcommand\orcontrol{q}

\newcommand\genop{o}
\newcommand\configc{c}
\newcommand\config[2]{\left\langle {#1}, {#2} \right\rangle}

\newcommand\stackw{w}
\newcommand\stacku{u}
\newcommand\stackv{v}

\newcommand\ctop[1]{top_{#1}}
\newcommand\cbottom[2]{bot^{#2}_{#1}}
\newcommand\cpdstran{\longrightarrow}

\newcommand\alphabet\Sigma
\newcommand\chorder{\lambda}
\newcommand\orchorder{\lambda'}

\newcommand\ccompose[3]{{#1} :_{#2} {#3}}

\newcommand\annot[2]{#1^{#2}}

\newcommand\sbrac[2]{[{#1}]_{#2}}
\newcommand\stacks[1]{{Stacks_{#1}}}

\newcommand\game{\mathcal{G}}

\newcommand\wincond{\mathcal{W}}
\newcommand\rankf{\rho}
\newcommand\orrankf{\orof{\rankf}}
\newcommand\crrankf{\crof{\rankf}}
\newcommand\prrankf{\prof{\rankf}}
\newcommand\rank{r}
\newcommand\owner{\mathbb{O}}
\newcommand\orowner{\orof{\owner}}

\newcommand\prowner{\prof{\owner}}
\newcommand\aplayer{\mathrm{A}}
\newcommand\eplayer{\mathrm{E}}
\newcommand\genplayer{\Upsilon}
\newcommand\strat{\sigma}
\newcommand\maxrank{m}
\newcommand\link{\ell}
\newcommand\colrankfun{\mathrm{Rk}}
\newcommand\higherranks{\mathrm{HRanks}}
\newcommand\allranks{\mathrm{Ranks}}
\newcommand\extrkf[2]{{#1}[#2]}

\newcommand\controlinit{\control_I}
\newcommand\stackinit{\stackw_I}
\newcommand\orcontrolinit{\orcontrol_I}

\newcommand\prcontrolinit{\prof{\control_I}}
\newcommand\chainit{\cha_I}
\newcommand\crchainit{\crof{\cha_I}}
\newcommand\crchorder{\crof{\chorder}}
\newcommand\prchorder{\prof{\chorder}}
\newcommand\prchainit{\prof{\cha_I}}

\newcommand\orof[1]{#1'}
\newcommand\orderred{\mathrm{Order}}
\newcommand\popguess{\vec{\controlset}}
\newcommand\undefpopguess{\circleddash}
\newcommand\ewin{\mathtt{tt}}
\newcommand\elose{\mathtt{ff}}
\newcommand\controlsink{\control_s}

\newcommand\numranks{M}
\newcommand\rankawareconst{N}
\newcommand\rankawarepower{D}

\newcommand\crof[1]{#1'}
\newcommand\counterred[1]{\mathrm{Counter}_{#1}}
\newcommand\bound{B}
\newcommand\counter{\alpha}
\newcommand\countervec{\vec{\counter}}
\newcommand\counterinc[1]{\oplus_{#1}}
\newcommand\overflow{\mathtt{NaN}}
\newcommand\eover{\#}
\newcommand\aover{\$}

\newcommand\prof[1]{#1'}

\newcommand\gameoc{\game_{O,C_\bound}}
\newcommand\gameco{\game_{C_\bound,O}}

\newcommand\gameo{\game_O}
\newcommand\gamec{\game_{C_\bound}}

\newcommand\linkalphabet{\Delta}
\newcommand\indexalphabet{\Gamma}
\newcommand\counteralphabet[1]{\indexalphabet_{#1}}
\newcommand\encodingalphabet[1]{\hat{\indexalphabet}_{#1}}
\newcommand\bzero[1]{0_{#1}}
\newcommand\bone[1]{1_{#1}}
\newcommand\countercheck[1]{\mathit{encoding}_{#1}}
\newcommand\equalscheck[1]{\mathit{equal}_{#1}}
\newcommand\polyred[1]{\mathrm{Poly}_{#1}}

\newcommand\linkchar[1]{\link_{#1}}
\newcommand\controlincnoa[2]{\tup{\mathrm{inc}, #1, #2}}
\newcommand\controlinc[3]{\tup{\mathrm{inc}, #1, #2, #3}}

\newcommand\controlzero[3]{\tup{\mathrm{zero}, #1, #2, #3}}
\newcommand\controlzerotest[3]{\tup{\mathrm{zchk}, #1, #2, #3}}
\newcommand\controlcopy[3]{\tup{\mathrm{copy}, #1, #2, #3}}
\newcommand\controlcopytest[3]{\tup{\mathrm{cchk}, #1, #2, #3}}
\newcommand\controlcollapse[3]{\tup{\mathrm{collapse}, #1, #2,  #3}}
\newcommand\controlpop[3]{\tup{\mathrm{pop}_1, #1, #2, #3}}
\newcommand\controlscw{\controls_{\mathit{CW}}}
\newcommand\controlsop{\controls_{\mathit{OP}}}
\newcommand\rulescw{\cpdsrules_{\mathit{CW}}}
\newcommand\gamebound{G_B}
\newcommand\gamepoly{G_P}

\makeatletter
\providecommand*{\dotcup}{%
    \mathbin{%
        \mathpalette\@dotcup{}%
    }%
}
\newcommand*{\@dotcup}[2]{%
    \ooalign{%
        $\m@th#1\cup$\cr
        \hidewidth$\m@th#1\cdot$\hidewidth
    }%
}
\makeatother

\newcommand\reftheorem[1]{\expandafter\csname reftheorem#1\endcsname}

\newenvironment{namedlemma}[2]{%
    \expandafter\gdef\csname reflemma#1\endcsname{%
        Lemma~\ref{#1} (#2)%
    }%
    \begin{lemma}[{#2}] \label{#1}%
}{%
    \end{lemma}%
}
\newcommand\reflemma[1]{\expandafter\csname reflemma#1\endcsname}

\newcommand\refproperty[1]{\expandafter\csname refproperty#1\endcsname}

\newenvironment{nameddefinition}[2]{%
    \expandafter\gdef\csname refdefinition#1\endcsname{%
        Definition~\ref{#1} (#2)%
    }%
    \begin{definition}[{#2}] \label{#1}%
}{%
    \end{definition}%
}
\newcommand\refdefinition[1]{\expandafter\csname refdefinition#1\endcsname}

\newtheorem{proposition}[theorem]{Proposition}

\title{Parity to Safety in Polynomial Time for Pushdown and Collapsible Pushdown Systems}
\titlerunning{Parity to Safety in Polynomial Time}

\author{Matthew Hague}
       {Royal Holloway, University of London}
       {matthew.hague@rhul.ac.uk}
       {https://orcid.org/0000-0003-4913-3800}
       {Supported by the EPSRC under grant [EP/K009907/1].}

\author{Roland Meyer}
       {TU Braunschweig}
       {roland.meyer@tu-braunschweig.de}
        {https://orcid.org/0000-0001-8495-671X}
       {}

\author{Sebastian Muskalla}
       {TU Braunschweig}
       {s.muskalla@tu-braunschweig.de}
       {https://orcid.org/0000-0001-9195-7323}
       {}

\author{Martin Zimmermann}
       {Universit\"at des Saarlandes}
       {zimmermann@react.uni-saarland.de}
       {https://orcid.org/0000-0002-8038-2453}
       {Supported by the DFG under grant [ZI 1516/1-1].}

\authorrunning{M. Hague, R. Meyer, S. Muskalla, M. Zimmermann}
\Copyright{Matthew Hague, Roland Meyer, Sebastian Muskalla, Martin Zimmermann}

\subjclass{\\
    \ccsdesc[500]{Theory of computation~Logic and verification},\\
    \ccsdesc[500]{Theory of computation~Modal and temporal logics},\\
    \ccsdesc[500]{Theory of computation~Verification by model checking},\\
    \ccsdesc[300]{Theory of computation~Grammars and context-free languages}.\\
}

\keywords{Parity Games, Safety Games, Pushdown Systems, Collapsible Pushdown Systems, Higher-Order Recursion Schemes, Model Checking}

\acknowledgements{We thank the anonymous reviewers for their comments.}

\begin{document}

\maketitle

\begin{abstract}
    We give a direct polynomial-time reduction from parity games played over the configuration graphs of collapsible pushdown systems to safety games played over the same class of graphs. That a polynomial-time reduction would exist was known since both problems are complete for the same complexity class. Coming up with a direct reduction, however, has been an open problem. Our solution to the puzzle brings together a number of techniques for pushdown games and adds three new ones. This work contributes to a recent trend of liveness to safety reductions which allow the advanced state-of-the-art in safety checking to be used for more expressive specifications.

\end{abstract}

\newpage

\section{Introduction}

Model-checking games ask whether there is a strategy (or implementation) of a system that can satisfy required properties against an adversary (or environment).
They give a natural method for reasoning about systems wrt. popular specification logics such as LTL, CTL, and the $\mu$-calculus.
The simplest specifications are reachability or safety properties, where the system either needs to reach a given good state or avoid a bad state (such as a null-pointer dereference).
The most expressive logic typically studied is the $\mu$-calculus, which subsumes LTL, CTL, and CTL$^\ast$~\cite{L05}.
One can reduce $\mu$-calculus model checking in polynomial time to the analysis of parity games (op cit.) via a quite natural product of system and formula.

In the finite-state setting, while reachability and safety games can be solved in linear time and space, the best known algorithms for parity games are quasi-polynomial time~\cite{CJKLS17} or quasi-linear space~\cite{JL17,FJSSW17}.
For infinite-state games described by pushdown systems, or more generally, collapsible pushdown systems, the complexities match: EXPTIME-complete for solving reachability, safety~\cite{BEM97,W96}, and parity games~\cite{W96} over pushdown systems, and $\cpdsord$-EXPTIME-complete for order-$\cpdsord$ collapsible pushdown systems~\cite{E91,CW07,O06,HMOS08}.

Pushdown systems are an operational model for programs with (recursive) function calls.
In such systems, a configuration has a control state from a finite set and a stack of characters from a finite alphabet (modeling the call stack).
Collapsible pushdown systems~\cite{HMOS08} are an operational model for higher-order recursion as found in most languages (incl. Haskell, JavaScript, Python, C++, Java, \ldots).
They have a nested stack-of-stacks (e.g.\ an order-$2$ stack is a stack of stacks) and \emph{collapse links} which provide access to calling contexts.

Given that safety and parity games over collapsible pushdown systems are complete for the same complexity classes, the problems must be inter-reducible in polynomial-time.
However, a direct (without a detour via Turing machines) polynomial-time reduction from parity to safety has been an open problem~\cite{FZ12}.
To see why the reduction is difficult to find, note that a safety game is lost based on a finite prefix of a play while determining the winner of a parity game requires access to the infinitely many elements of a play.
Complexity theory tells us that this gap can be bridged by access to the stack, with only polynomial overhead.

Our contribution is such a polynomial-time reduction from parity to safety.
From a theoretical standpoint, it explains the matching complexities despite the difference in expressible properties.
From a practical standpoint, it may help building model-checking tools for $\mu$-calculus specifications.
Indeed, competitive and highly optimized tools exist for analysing reachability and safety properties of higher-order recursion schemes (HorSat~\cite{BK13,TK14,horsat2} and Preface~\cite{RNO14} being the current state-of-the-art), but implementing efficient tools for parity games remains a problem~\cite{FIK13,NO14}.
Having the reduction at hand can allow the use of safety tools for checking parity conditions, suggest the transfer of techniques and optimizations from safety to parity, and inspire new algorithms for parity games.
Still, a complexity-theoretic result should only be considered a first step towards practical developments.

Reductions from parity to safety have been explored for the finite-state case by Bernet\etal~\cite{BJW02}, and for pushdown systems by Fridman and Zimmermann~\cite{FZ12}.
We will refer to them as counter reductions, as they use counters to track the occurrences of odd ranks.
These existing reductions are not polynomial.
Berwanger and Doyen~\cite{BD08} showed that counter reductions can be made polynomial in the case of finite-state imperfect-information games.

Our solution to the puzzle brings together a number of techniques for pushdown games and contributes three new ones.
We first show how to lift the existing counter reductions~\cite{BJW02,FZ12} from first order to higher orders.
For this we exploit a rank-awareness property of collapsible pushdown systems~\cite{HMOS08}.
Secondly, we prove the correctness of this lifting by showing that it commutes with a reduction from order-$\cpdsord$ to order-$(\cpdsord-1)$ games~\cite{W96,HMOS08}.
The polynomial-time reduction is then a compact encoding of the lifted counter reduction.
It uses the ability of higher-order stacks to encode large numbers~\cite{CW07} and the insight that rank counters have a stack-like behavior, even in their compact encoding.

Much recent work verifies liveness properties via reductions to safety ~\cite{PHLPSS18,DCGTM16,PR04,PR11,FKP16} or reachability~\cite{KLVW17,BJKKRVW16} with promising results.
For finite-state generalized parity games, Sohail and Somenzi show that pre-processing via a safety property can reduce the state space that a full parity algorithm needs to explore, giving competitive synthesis results for LTL~\cite{SS09}.
In the case of infinite-state systems (including pushdowns), reductions from liveness (but not parity games) have been explored by Biere\etal~\cite{BAS02} and Schuppan and Biere~\cite{SB05}.

\section{Preliminaries}%
\label{sec:preliminaries}
We define games over collapsible pushdown systems (CPDS).
For a full introduction see~\cite{HMOS08}.
CPDS are an operational model of functional programs that is equivalent to higher-order recursion schemes (HORS)~\cite{HMOS08}.
Without collapse, they correspond to \emph{safe} HORS~\cite{KNU02}.

In the following,
let $\naturals$ be the set of natural numbers (including $0$)
and
$[\idxi,\idxj]$
denote the set
$\set{\idxi, \idxi + 1, \ldots, \idxj}$.

\subsection{Higher-Order Collapsible Stacks}
Higher-order stacks are a nested stack-of-stacks structure whose stack characters are annotated by collapse links that point to a position in the stack.
Intuitively, this position is the context in which the character was created.
We describe the purpose of collapse links after some basic definitions.
\begin{definition}[Order-$\cpdsord$ Collapsible Stacks]
    For $\cpdsord \geq 1$, let $\alphabet$ be a finite set of stack characters $\alphabet$ together with a partition function\footnote{%
        Readers familiar with CPDS may expect links to be pairs
        $\tup{\opord, \idxi}$
        and the alphabet $\alphabet$ not to be partitioned by link order.
        The partition assumption is oft-used.
        It is always possible to tag each character with its link order using
        $\alphabet \times [1,\cpdsord]$.
        Such a partition becomes crucial in Section~\ref{sec:orderreduction}.
    }
    $\chorder : \alphabet \rightarrow [1,\cpdsord]$.
    An \emph{order-$0$ stack} with up-to order-$\cpdsord$ collapse links is an annotated character
    $\cha^\idxi \in \alphabet \times \naturals$.
    An \emph{order-$\opord$ stack} with up-to order-$\cpdsord$ collapse links is a non-empty sequence
    $\stackw = \sbrac{\stackw_1 \ldots \stackw_\numof}{\opord}$
    (with $\numof > 0$)
    such that each $\stackw_\idxi$ is an order-$(\opord-1)$ stack with up-to order-$\cpdsord$ collapse links.
    By $\stacks{\cpdsord}$ we denote the set of order-$\cpdsord$ stacks with up-to order-$\cpdsord$ links.
\end{definition}
In the sequel, we will refer to stacks in $\stacks{\cpdsord}$ as order-$\cpdsord$ stacks.
By \emph{order-$\opord$} stack we will mean an order-$\opord$ stack with up-to order-$\cpdsord$ links, where $\cpdsord$ will be clear from the context.

Given an order-$\opord$ stack with up-to order-$\cpdsord$ links
$\stackw=\sbrac{\stackw_1\ldots\stackw_\numof}{\opord}$, we define below the operation $\ctop{\opord'}$ to return the topmost element of the topmost order-$\opord'$ stack.
Note that this element is of order-($\opord'-1$).
The top of a stack appears leftmost.
The operation $\cbottom{\opord}{\idxi}$ removes all but the last $\idxi$ elements from the topmost order-$\opord$ stack.
It does not change the order of the stack and requires $\idxi \in [1, \ell]$.
\[
    \begin{array}{rclcrcl}
        \ap{\ctop{\opord}}{\stackw}
        &=&
        \stackw_1

        &\quad&

        \ap{\cbottom{\opord}{\idxi}}{\stackw}
        &=&
        \sbrac{
            \stackw_{\numof - \idxi + 1} %
            \ldots
            \stackw_\numof
        }{\opord} \\

        \ap{\ctop{\opord'}}{\stackw}
        &=&
        \ap{\ctop{\opord'}}{\stackw_1} \text{ ($\opord' < \opord$)}

        &\quad&

        \ap{\cbottom{\opord'}{\idxi}}{\stackw}
        &=&
        \sbrac{
            \ap{\cbottom{\opord'}{\idxi}}{\stackw_1}
            \stackw_2 \ldots \stackw_\numof
        }{\opord}
        \text{ ($\opord' < \opord$).}
    \end{array}
\]
For technical convenience, we will also define
\[
    \ap{\ctop{\cpdsord+1}}{\stackw} = \stackw
\]
which, we note, does not extend to $\ctop{\cpdsord+2}$ or beyond.

The destination of a collapse link $\idxi$ on $\cha$ with $\ap{\chorder}{\cha} = \opord$ in a stack $\stackw$
is
$\ap{\cbottom{\opord}{\idxi}}{\stackw}$,
when defined.
When $\idxi=0$, the link is considered \emph{null}.
We often omit irrelevant collapse links from characters to improve readability.


When $\stacku$ is a $(\opord-1)$-stack and
$\stackv=\sbrac{\stackv_1\ldots\stackv_\numof}{\cpdsord}$
is an $\cpdsord$-stack with $\opord\in [1, \cpdsord]$, we define
$\ccompose{\stacku}{\opord}{\stackv}$
as the stack obtained by adding $\stacku$ on top of the topmost $\opord$-stack of $\stackv$.
Formally,
\[
    \ccompose{\stacku}{\opord}{\stackv} =
        \sbrac{\stacku\stackv_1\ldots\stackv_\numof}{\cpdsord}
    \text{ ($\opord= \cpdsord$)}
    \qquad
    \text{and}
    \qquad
    \ccompose{\stacku}{\opord}{\stackv} =
        \sbrac{(\ccompose{\stacku}{\opord}{\stackv_1})\stackv_2\ldots\stackv_\numof}{\cpdsord}
    \text{ ($\opord< \cpdsord$).}
\]

\begin{example} \label{eg:stack-ops}
    When
    $\ap{\chorder}{\cha} = 3$
    and
    $\ap{\chorder}{\chb} = 2$
    let
    $\stackw =
        \sbrac{
            \sbrac{
                \sbrac{\annot{\cha}{1} \annot{\chb}{1}}{1}
                \sbrac{\annot{\chb}{1}}{1}
            }{2}
            \sbrac{\sbrac{\annot{\chb}{0}}{1}}{2}
         }{3}$
    be an order-$3$ collapsible stack.
    The destination of the topmost link is
    $\ap{\cbottom{3}{1}}{\stackw}=\sbrac{
        \sbrac{\sbrac{\annot{\chb}{0}}{1}}{2}
    }{3}$.
    Furthermore,
    $\ap{\cbottom{2}{1}}{\stackw}=\sbrac{
        \sbrac{
            \sbrac{\annot{\chb}{1}}{1}
        }{2}
        \sbrac{\sbrac{\annot{\chb}{0}}{1}}{2}
    }{3}$
    and
    $\ap{\ctop{2}}{\stackw}=\sbrac{\annot{\cha}{1} \annot{\chb}{1}}{1}$.
    Here, $\ccompose{\ap{\ctop{2}}{\stackw}}{2}{\ap{\cbottom{2}{1}}{\stackw}}= \stackw$.
\end{example}
\subsubsection*{Operations on Order-$\cpdsord$ Collapsible Stacks}

CPDS are programs with a finite control acting on collapsible stacks via the operations:
\[
    \cops{\cpdsord}
    =
    \set{\push{2},\ldots,\push{\cpdsord}} \cup
    \setcomp{\cpush{\cha}, \rew{\cha}}
            {\cha \in \alphabet} \cup
    \set{\pop{1},\ldots,\pop{\cpdsord}} \cup
    \set{\collapse}.
\]
Operations $\push{\opord}$ of order $\opord>1$ copy the topmost element of the topmost order-$\opord$ stack.
Order-$1$ push operations $\cpush{\cha}$ push $\cha$ onto the topmost order-$1$ stack and annotate it with an order-$\ap{\chorder}{\cha}$ collapse link.
When executed on a stack $\stackw$, the link destination is $\ap{\pop{\ap{\chorder}{\cha}}}{\stackw}$.
A $\pop{\opord}$ removes the topmost element from the topmost order-$\opord$ stack.
The rewrite $\rew{\cha}$ modifies the topmost stack character while maintaining the link (rewrite must respect the link order).
Collapse, when executed on $\annot{\cha}{\idxi}$ with $\ap{\chorder}{\cha} = \opord$, pops the topmost order-$\opord$ stack down to the last $\idxi$ elements, captured by $\cbottom{\opord}{\idxi}$.
Formally, for an order-$\cpdsord$ stack $\stackw$:
\begin{enumerate}
\item
    $\ap{\push{\opord}}{\stackw} =
        \ccompose{\ap{\ctop{\opord}}{\stackw}}{\opord}{\stackw}$.

\item
    $\ap{\cpush{\cha}}{\stackw} =
        \ccompose{\cha^{\numof-1}}{1}{\stackw}$
    when $\ap{\ctop{\opord+1}}{\stackw}
        = \sbrac{\stackw_1\ldots\stackw_\numof}{\opord}$,
    where $\opord=\ap{\chorder}{\cha}$ is the link order,

\item
    $\ap{\pop{\opord}}{\stackw} = \stackv$
    when
    $\stackw = \ccompose{\stacku}{\opord}{\stackv}$,

\item
    $\ap{\collapse}{\stackw} =
        \ap{\cbottom{\opord}{\idxi}}{\stackw}$
    when
    $\ap{\ctop{1}}{\stackw} = \cha^{\idxi}$
    and
    $\ap{\chorder}{\cha} = \opord$,
    and

\item
    $\ap{\rew{\chb}}{\stackw} =
        \ccompose{\annot{\chb}{\idxi}}{1}{\stackv}$
    when
    $\stackw = \ccompose{\annot{\cha}{\idxi}}{1}{\stackv}$
    and
    $\ap{\chorder}{\cha} = \ap{\chorder}{\chb}$.
\end{enumerate}

\noindent
Note that since our definition of stacks does not permit empty stacks, $\pop{\opord}$ is undefined if $\stackv$ is empty and $\collapse$ is undefined when $\idxi = 0$.
Thus, the empty stack cannot be reached using CPDS operations;
instead, the offending operation will simply be unavailable.
Likewise if a rewrite operation would change the order of the link.
\begin{example}
    Recall Example~\ref{eg:stack-ops} and that
    $\stackw =
    \sbrac{
        \sbrac{
            \sbrac{\annot{\cha}{1} \annot{\chb}{1}}{1}
            \sbrac{\annot{\chb}{1}}{1}
        }{2}
        \sbrac{\sbrac{\annot{\chb}{0}}{1}}{2}
     }{3}$.
    Given the order-$3$ link $1$ of the topmost stack character $\cha$, a collapse operation yields
    $\stacku=\sbrac{
        \sbrac{\sbrac{\annot{\chb}{0}}{1}}{2}
    }{3}$.
    Now
    $\ap{\push{3}}{\stacku}=\sbrac{
        \sbrac{\sbrac{\annot{\chb}{0}}{1}}{2}
        \sbrac{\sbrac{\annot{\chb}{0}}{1}}{2}
         }{3}$.
    A $\cpush{\cha}$ on this stack results in
    $\stackv = \sbrac{\sbrac{
        \sbrac{\annot{\cha}{1}\annot{\chb}{0}}{1}}{2}
        \sbrac{\sbrac{\annot{\chb}{0}}{1}}{2}
    }{3}$.
    We have $\ap{\pop{3}}{\stackv}=\stacku=\ap{\collapse}{\stackv}$.
\end{example}

There is a subtlety in the interplay of collapse links and higher-order pushes.
For a $\push{\opord}$, links pointing outside of
$\stacku = \ap{\ctop{\opord}}{\stackw}$
have the same destination in both copies of $\stacku$, while links pointing within $\stacku$ point to different sub-stacks.

\begin{remark}[$\noop$]
    For convenience we use an operation $\noop$ which has no effect on the stack.
    We can simulate it by $\rew{\cha}$ where $\cha$ is the topmost character
    (by the format of rules, below, we will always know the topmost character when applying an operation).
    Hence, it is not a proof case.
\end{remark}
\subsection{Collapsible Pushdown Systems and Games}

\begin{definition}[CPDS]
An order-$\cpdsord$ \emph{collapsible pushdown system} is a tuple $\cpds$ given by
$\tup{\controls, \alphabet, \chorder, \cpdsrules,
      \controlinit, \chainit, \rankf}$
with
    $\controls$ a finite set of control states with
    initial control state $\controlinit$,
    $\alphabet$ a finite stack alphabet with initial stack character
        $\chainit$
        and order function $\chorder$,
    $\rankf:\controls\rightarrow\naturals$
        a function assigning ranks to $\controls$, and
    $\cpdsrules \subseteq
            \brac{\controls \times \alphabet \times \cops{\cpdsord} \times \controls}$
     a set of rules.
     The size is
     $\sizeof{\cpds}=\sizeof{\controls}+\sizeof{\alphabet}$.
    The remaining entries polynomially depend on $\controls$ and $\alphabet$
    (note that $\cpdsord$ is fixed).
\end{definition}
We attach ranks to CPDS instead of games as we later need the notion of \emph{rank-aware} CPDS.

A \emph{configuration} of a CPDS is a pair
$\configc=\config{\control}{\stackw}$ with $\control \in \controls$ and a stack $\stackw \in \stacks{\cpdsord}$.
We have a transition
$\config{\control}{\stackw} \cpdstran \config{\control'}{\stackw'}$
if there is a rule
$\cpdsrule{\control}{\cha}{\genop}{\control'}\in \cpdsrules$
with
$\ap{\ctop{1}}{\stackw} = \cha$
and
$\stackw' = \ap{\genop}{\stackw}$.
The initial configuration is $\config{\controlinit}{\stackinit}$
where
$\stackinit = \sbrac{\ldots\sbrac{\annot{\chainit}{0}}{1}\ldots}{\cpdsord}$.
To begin from another configuration, one can adjust the CPDS rules to build the required stack from the initial configuration.
A computation is a sequence of configurations $\configc_0, \configc_1, \ldots$ where $\configc_0=\config{\controlinit}{\stackinit}$ and $\configc_i\cpdstran \configc_{i+1}$ for all $i\in \naturals$.
Recall, transitions cannot empty a stack or rewrite the order of a link.
\begin{definition}[Games over CPDS]
    A \emph{game over a CPDS} is a tuple
    $\game = \tup{\cpds, \owner, \wincond}$,
    where $\cpds$ is a CPDS,
        $\owner : \controls \rightarrow \set{\aplayer,\eplayer}$
            is a division of the control states of $\cpds$ by owner \theeplayer ($\eplayer$) or \theaplayer ($\aplayer$), and
        $\wincond \subseteq \naturals^{\omega}$
            is a winning condition.
            The size of the game is $\sizeof{\game}=\sizeof{\cpds}$.

    We call $\game$ a \emph{safety game} if
            $\ap{\rankf}{\control} \in \set{1,2}$
            for all
            $\control \in \controls$
            and
            $\wincond = 2^\omega$.
            It is
            a \emph{parity game} if $\wincond$ is the set of all sequences such that the smallest infinitely occurring rank is even.
\end{definition}

We refer to computations as plays and require them to be infinite.
This means every configuration
$\config{\control}{\stackw}$
has some successor
$\config{\control'}{\stackw'}$.
This does not lose generality as we can add to the CPDS transitions to a losing (as defined next) sink state (with self-loop) for
$\ap{\owner}{\control}$
from any configuration
$\config{\control}{\stackw}$.
A play $\config{\control_0}{\stackw_0},
    \config{\control_1}{\stackw_1},
    \config{\control_2}{\stackw_2},
    \ldots$
is won by \theeplayer, if its sequence of ranks satisfies the winning condition, i.e.\ %
$
    \ap{\rankf}{\control_0}
    \ap{\rankf}{\control_1}
    \ap{\rankf}{\control_2}
    \ldots
    \in \wincond.
$
Otherwise, \theaplayer wins.
When a play reaches $(\control, \stackw)$, then the owner of $\control$ chooses the rule to apply.
A \emph{strategy} for player
$\genplayer \in \set{\eplayer, \aplayer}$
is a function
$
    \strat : \brac{\controls \times \stacks{\cpdsord}}^\ast
             \rightarrow
             \cpdsrules \
$
that returns an appropriate rule based on the prefix of the play seen so far.
A play
$
    \config{\control_0}{\stackw_0},
    \config{\control_1}{\stackw_1},
    \config{\control_2}{\stackw_2},
    \ldots
$
is according to $\strat$ if for all $\idxi$ with
$\ap{\owner}{\control_\idxi} = \genplayer$
we have
$\config{\control_\idxi}{\stackw_\idxi}
 \cpdstran
 \config{\control_{\idxi+1}}{\stackw_{\idxi+1}}$
via rule
$\ap{\strat}{
    \config{\control_0}{\stackw_0},
    \ldots,
    \config{\control_\idxi}{\stackw_\idxi}
 }$.
The strategy is winning if all plays according to $\strat$ are won by $\genplayer$.
We say a player \emph{wins} a game if they have a winning strategy from the initial configuration.

\subsection{Rank-Aware Collapsible Pushdown Systems}

We will often need to access the smallest rank that was seen in a play since some stack was created.
\emph{Rank-aware} CPDS record precisely this information~\cite{HMOS08}.
We first define $\opord$-ancestors which, intuitively, give the position in the play where the top order-$(\opord-1)$ stack was pushed.
Note, in the definition below, the integer $\idxj$ is unrelated to the collapse links.
\begin{definition}[$\opord$-Ancestor]
    Let $\opord \in [2, \cpdsord]$ (resp.\ $\opord = 1$).
    Given a play $\configc_0, \configc_1, \ldots$
    we attach an integer $\idxj$ to every order-$(\opord-1)$ stack as follows.
    In $\configc_0$ all order-$(\opord-1)$ stacks are annotated by $0$.
    Suppose $\configc_{\idxi+1}$ was obtained from $\configc_\idxi$ using operation
    $\push{\opord}$
    (resp.\ $\cpush{\cha}$).
    Then the new topmost order-$(\opord-1)$ stack in $\configc_{\idxi+1}$ is annotated with $\idxi$.
    If $\configc_{\idxi+1}$ is obtained via a $\push{\opord''}$ with $\opord'' > \opord$, then all annotations on the order-$(\opord-1)$ stacks in the copied stack are also copied.

    The \emph{$\opord$-ancestor} with $\opord\in [1, \cpdsord]$ of $\configc_\idxi$ is the configuration $\configc_\idxj$ where $\idxj$ is the annotation of the topmost order-$(\opord-1)$ stack in $\configc_\idxi$.
    Let
    $\ap{\ctop{1}}{\configc_{\idxi}}=\cha^{\numof}$
    and
    $\ap{\chorder}{\cha} = \opord'$.
    The \emph{link ancestor} of $\configc_\idxi$  is the $\opord'$-ancestor of the $1$-ancestor of $\configc_{\idxi}$.
\end{definition}
Applying a $\pop{\opord}$ operation will expose (a copy of) the topmost $(\opord-1)$-stack of the $\opord$-ancestor.
To understand the notion of a link ancestor, remember that $\collapse$ executed on a stack whose topmost order-$0$ stack is
$\annot{\cha}{\numof}$
with
$\ap{\chorder}{\cha} = \opord'$
has the effect of executing $\pop{\opord'}$ several times.
The newly exposed topmost $(\opord'-1)$-stack is the same that would be exposed if $\pop{\opord'}$ were applied at the moment the $\cha$ character was pushed.
This exposed stack is the same stack as is topmost on the $\opord'$-ancestor of the $1$-ancestor of $\cha$.
We illustrate this with an~example.

\begin{example}
    Assume some $\configc_0$.
    Now take some $\configc_1$ containing the stack
    $\stackw_1 = \sbrac{\sbrac{\sbrac{\annot{\chb}{0}}{1}}{2}}{3}$.
    Apply a $\push{3}$ operation to obtain $\configc_2$ with stack
    $\stackw_2 = \sbrac{
        \sbrac{\sbrac{\annot{\chb}{0}}{1}}{2}
        \sbrac{\sbrac{\annot{\chb}{0}}{1}}{2}
    }{3}$.
    Note, the topmost
    $\sbrac{\sbrac{\annot{\chb}{0}}{1}}{2}$ has $3$-ancestor $\configc_1$.

    Now, let
    $\ap{\chorder}{\cha} = 3$ and
    obtain $\configc_3$ with $\cpush{\cha}$, which thus contains the stack
    $\stackw_3 =
    \sbrac{
        \sbrac{\sbrac{\annot{\cha}{1} \annot{\chb}{0}}{1}}{2}
        \sbrac{\sbrac{\annot{\chb}{0}}{1}}{2}
     }{3}$
    where the $\annot{\cha}{1}$ has the $1$-ancestor $\configc_2$.

    We can now apply $\push{3}$ again to reach $\configc_4$ with stack
    $\stackw_4 =
    \sbrac{
        \sbrac{\sbrac{\annot{\cha}{1} \annot{\chb}{0}}{1}}{2}
        \sbrac{\sbrac{\annot{\cha}{1} \annot{\chb}{0}}{1}}{2}
        \sbrac{\sbrac{\annot{\chb}{0}}{1}}{2}
     }{3}$.
    Note that both copies of $\annot{\cha}{1}$ have the $1$-ancestor $\configc_2$.
    Moreover, the link ancestor of both is $\configc_1$.  That is, the $3$-ancestor of the topmost stack of $\configc_2$.
    In particular, applying $\collapse$ at $\configc_4$ results in a configuration with stack
    $\sbrac{\sbrac{\sbrac{\annot{\chb}{0}}{1}}{2}}{3}$,
    which is the same stack contained in $\configc_1$.
\end{example}

In the below, intuitively, the level-$\opord$ rank is the smallest rank seen since the topmost $(\opord-1)$ stack was created.
Similarly for the link-level rank.
Our rank-awareness definition is from~\cite{HMOS08} but includes level-$\opord$ ranks as well.
\begin{definition}[Level Rank]
    For a given play
    $\configc_0, \configc_1, \ldots$
    the \emph{level-$\opord$ rank} with $\opord\in [1, \cpdsord]$ (resp.
    \emph{link-level rank}) at a configuration $\configc_\idxi$ is the smallest rank of a control state in the sequence
    $\configc_{\idxj+1},\ldots,\configc_\idxi$
    where $\idxj$ is the $\opord$-ancestor (link ancestor) of $\configc_\idxi$.
\end{definition}

In the following definition, $\link$ is a special symbol to be read as \emph{link}.

\begin{nameddefinition}{def:rank-aware}{Rank-Aware}
    A \emph{rank-aware CPDS} is a CPDS $\cpds$ over stack characters $\tup{\cha, \colrankfun}$, where $\cha$ is taken from a finite set and function $\colrankfun$ has type
    $\colrankfun : {[1, \cpdsord]\dotcup\set{\link}} \rightarrow [0, \maxrank]$
    (with $\maxrank$ the highest rank of a state in $\cpds$).
    The requirement is that in all computations $\configc_0, \configc_1, \ldots$ of the CPDS
    all configurations $\configc_\idxi=\tup{\control, \stackw}$ with top-of-stack character $\tup{\cha, \colrankfun}$ satisfy
    \[
        \ap{\colrankfun}{\opord} = \text{ the level-$\opord$ rank at $\configc_\idxi$, $\opord\in[1, \cpdsord]$},
        \quad
        \text{and}
        \quad
        \ap{\colrankfun}{\link} = \text{ the link-level rank at $\configc_\idxi$.}
    \]
\end{nameddefinition}
Below, we slightly generalize a lemma from~\cite{HMOS08} to include safety games and level-$\opord$ ranks.
Intuitively, we can obtain rank-awareness by keeping track of the required information in the stack characters and control states.
\begin{namedlemma}{lem:rank-aware}{Rank-Aware}
    Given a parity (resp.\ safety) game over CPDS $\cpds$ of order-$\cpdsord$,
    one can construct in polynomial time a rank-aware CPDS $\cpds'$ of the same order and a parity (resp.\ safety) game over $\cpds'$ such that \theeplayer wins the game over $\cpds$ iff he wins the game over $\cpds'$.
\end{namedlemma}
Note, the number of functions $\colrankfun$ is exponential in $\cpdsord$.
However, since $\cpdsord$ is fixed the construction is polynomial.
In the sequel, we will assume that all CPDS are rank-aware.

\section{Main Result and Proof Outline}
\label{sec:outline}
In the following sections, we define a reduction $\polyred{}$ which takes a parity game $\game$ over a CPDS and returns a safety game
$\ap{\polyred{}}{\game}$ of the same order.
The main result follows.
\begin{theorem}[From Parity to Safety, Efficient]
Given a parity game $\game$, $\theeplayer$ wins $\game$ iff he wins $\ap{\polyred{}}{\game}$.
$\ap{\polyred{}}{\game}$ is polynomially large and computable in time polynomial in the size of $\game$.
\end{theorem}
We outline how to define $\polyred{}$ and prove it correct.
First, we give a function $\counterred{\bound}$ reducing an order-$\cpdsord$ parity game to an equivalent order-$\cpdsord$ safety game.
It extends Fridman and Zimmermann's reduction~\cite{FZ12} from first order to higher orders.
In the finite-state setting, a related reduction appeared already in~\cite{BJW02}.
The idea is to count in the stack characters the occurrences of odd ranks.
\theeplayer has to keep the counter values below $\bound$, a threshold that is a parameter of the reduction.
For completeness, this threshold has to be $n$-fold exponential in the size of $\game$.
Let
$\ap{\nexp{0}}{f} = f$
and
$\ap{\nexp{\cpdsord}}{f} = 2^{\ap{\nexp{\cpdsord-1}}{f}}$.
We have the following lemma.
\begin{namedlemma}{lem:parityvssafety}{From Parity to Safety, Inefficient}
    Given a parity game $\game$ played over an order-$\cpdsord$ CPDS, there is a bound
    $\ap{\bound}{\game} = \ap{\nexp{\cpdsord}}{\ap{f}{\sizeof{\game}}}$
    for some polynomial $f$ so that for all $\bound \geq \ap{\bound}{\game}$ \theeplayer wins $\game$ iff he wins the safety game
    $\ap{\counterred{\bound}}{\game}$.
\end{namedlemma}
The size of
$\ap{\counterred{\bound}}{\game}$
is not polynomial, even for constant~$\bound$.
The next step is to give an efficient reduction $\polyred{\bound}$
producing a safety game equivalent to $\ap{\counterred{\bound}}{\game}$.
In particular $\ap{\polyred{\ap{\bound}{\game}}}{\game}$ can be computed in time polynomial only in the size of $\game$, not in $\ap{\bound}{\game}$.
Thus, we can define $\polyred{}$ from the main theorem to be
$\ap{\polyred{}}{\game}=\ap{\polyred{\ap{\bound}{\game}}}{\game}$.

Technically, $\polyred{}$ relies on the insight that counter increments as performed by $\counterred{\bound}$ follow a stack discipline.
Incrementing the $\rank$th counter resets all counters for $\rank' > \rank$ to zero.
The upper bound combines this with the fact that collapsible pushdown systems can encode large counters~\cite{CW07}.
The second step is summarized as follows.
\begin{namedlemma}{lem:poly-equiv}{From Inefficient to Efficient}
    \theeplayer wins
    $\ap{\counterred{\bound}}{\game}$
    iff he wins
    $\ap{\polyred{\bound}}{\game}$.
    Moreover,
    $\ap{\polyred{}}{\game}=\ap{\polyred{\ap{\bound}{\game}}}{\game}$
    is polynomial-time computable.
\end{namedlemma}
It should be clear that the above lemmas, once proven, yield the main theorem.
For the equivalence stated there, note that $\ap{\polyred{}}{\game}=\ap{\polyred{\ap{\bound}{\game}}}{\game}$ is equivalent to the game $\ap{\counterred{\ap{\bound}{\game}}}{\game}$ by Lemma~\ref{lem:poly-equiv}.
This game, in turn, is equivalent to $\game$ by Lemma~\ref{lem:parityvssafety}.

The proof of Lemma~\ref{lem:poly-equiv} will be direct and is given in Section~\ref{sec:polynomial}.
We explain the proof of Lemma~\ref{lem:parityvssafety} here, which relies on a third reduction.
We define a function called $\orderred$ that takes an order-$\cpdsord$ parity or safety game and produces an equivalent order-($\cpdsord-1$) parity or safety game.
The reduction already appears in~\cite{HMOS08}, and generalizes the one from~\cite{W96}.
Let
\[
    \begin{array}{llllll}
        \gameo &=& \ap{\orderred}{\game}, &
        \gamec &=& \ap{\counterred{\bound}}{\game}, \\
        \gameoc &=& \ap{\counterred{\bound}}{\ap{\orderred}{\game}}, &
        \gameco &=& \ap{\orderred}{\ap{\counterred{\bound}}{\game}} \ . \\
    \end{array}
\]
The proof of Lemma~\ref{lem:parityvssafety} chases the diagram below.
We rely on the observation that the games
$\ap{\counterred{\bound}}{\ap{\orderred}{\game}}$
and
$\ap{\orderred}{\ap{\counterred{\bound}}{\game}}$
are equivalent, as stated in Lemma~\ref{lem:commutes}.
The proof of Lemma~\ref{lem:commutes} needs the reductions and can be found in Section~\ref{sec:equivalence}.
The commutativity argument yields the following proof, almost in category-theoretic style.

\begin{center}
\begin{tikzpicture}[x=40ex,y=11ex, every node/.style={fill=white},semithick,>=stealth];
    \node (G) at (0,1) {$\game$};
    \node (C) at (1,1) {$\gamec$};
    \node (O) at (0,0) {$\gameo$};
    \node (OC) at (.7,0) {$\gameoc$};
    \node (CO) at (1,0) {$\gameco$};
    \draw [->, style=dashed] (G) -- node {$\counterred{\bound}$} (C);
    \draw [->] (G) -- node {$\orderred$} (O);
    \draw [->] (O) -- node {$\counterred{\bound}$} (OC);
    \draw [->] (C) -- node {$\orderred$} (CO);
    \draw [implies-implies,double equal sign distance] (CO) -- (OC);
\end{tikzpicture}
\end{center}

\begin{namedlemma}{lem:commutes}{$\gameoc$ vs. $\gameco$}
    Given $\bound \in \naturals$ and a parity game $\game$ over an order-$\cpdsord$ CPDS,
    \theeplayer wins $\ap{\counterred{\bound}}{\ap{\orderred}{\game}}$ iff \theeplayer wins $\ap{\orderred}{\ap{\counterred{\bound}}{\game}}$.
\end{namedlemma}

\begin{proof}[Proof of Lemma~\ref{lem:parityvssafety}]
    We induct on the order.
    At order-$1$, the result is due to Fridman and Zimmerman~\cite{FZ12}.
    For the induction, without the bound, at order-$\cpdsord$, take a winning strategy for \theeplayer in $\game$.
    By~\cite{HMOS08}, he has a winning strategy in $\gameo$.
    By induction, \theeplayer has a winning strategy in $\gameoc$ and by Lemma~\ref{lem:commutes} also in $\gameco$ when $\bound$ is suitably large.
    Finally, again by~\cite{HMOS08}, \theeplayer can win $\gamec$.
    I.e., we chase the diagram above from
    $\game$ to $\gameo$ to $\gameoc$ to $\gameco$ and then up to $\gamec$.
    To prove the opposite direction, simply follow the path in reverse.

    We obtain the required bound in Appendix~\ref{sec:bound-calc}.
    Intuitively, we have an exponential bound at order-$1$ by Fridman and Zimmerman.
    Thus, assume by induction we have a $(\cpdsord-1)$-fold exponential bound for order-$(\cpdsord-1)$.
    From an order-$\cpdsord$ system we obtain an exponentially large order-$(\cpdsord-1)$ system for which an $\cpdsord$-fold exponential bound is therefore needed.
\end{proof}
In Sections~\ref{sec:orderreduction} and~\ref{sec:counterreduction}, we define $\orderred$ and $\counterred{\bound}$, and show Lemma~\ref{lem:commutes} in Section~\ref{sec:equivalence}.
The reduction $\polyred{}$ is defined in Section~\ref{sec:polynomial}, which also sketches the proof of Lemma~\ref{lem:poly-equiv}.

\section{Order Reduction} 
\label{sec:orderreduction}

We recall the reduction of~\cite{HMOS08} from order-$\cpdsord$ to order-$(\cpdsord-1)$ parity games.
This reduction also works for safety games.
It is a natural extension of Carayol\etal~\cite{CHMOS08} for higher-order pushdown systems without collapse, which extended Walukiewicz's reduction of pushdown parity games to finite-state parity games~\cite{W96}.
Due to space constraints, we only give the intuition here.
It is useful when explaining the motivation behind the constructions in our parity to safety reduction.
The full construction is in Appendix~\ref{sec:appendix-orderreduction}.

Given an order-$\cpdsord$ CPDS $\cpds$ and a game 
$\game = \tup{\cpds, \owner, \wincond}$
we define an order-$(\cpdsord-1)$ game
$\ap{\orderred}{\game}$
over a CPDS $\orcpds$.
The order-$(\cpdsord-1)$ CPDS $\orcpds$ simulates $\cpds$.
The key operations are
    $\push{\cpdsord}$, 
    $\pop{\cpdsord}$, 
    $\cpush{\cha}$ with $\ap{\chorder}{\cha} = \cpdsord$,
        and
    $\collapse$ when the link is order-$\cpdsord$.
We say these operations are order-$\cpdsord$. 
The remaining operations are simulated directly on the stack of $\cpds'$.

There is no $\push{\cpdsord}$ on an order-\mbox{$(\cpdsord-1)$} stack.
Instead, observe that if the stack is $\stackw$ before the $\push{\cpdsord}$ operation, it will return to $\stackw$ after the corresponding $\pop{\cpdsord}$ (should it occur).
Thus, we simulate $\push{\cpdsord}$ by splitting the play into two branches.
The first simulates the play between the $\push{\cpdsord}$ and corresponding $\pop{\cpdsord}$.  The second simulates the play after the $\pop{\cpdsord}$.

Instead of applying a $\push{\cpdsord}$ operation, \theeplayer makes a claim about the control states the play may pop to.
Also necessary is information about the smallest rank seen in the play to the pop. 
This claim is recorded as a vector of sets of control states
$\popguess = \tup{\controlset_0, \ldots, \controlset_\maxrank}$ which is held
in the current control state.
Each 
$\control \in \controlset_\rank$
is a potential future of the play, meaning that the pushed stack may be popped to $\control$ and the minimum rank seen since the push could be $\rank$.
Because \theeplayer does not have full control of the game, he cannot give a single control state and rank: \theaplayer may force him to any of a number of situations.

Once this guess has been made, \theaplayer chooses whether to simulate the first play (between the push and the pop) or the second (after the pop).
In the first case, $\popguess$ is stored in the control state.
Then, when the pop occurs, \theeplayer wins if the destination control state is in $\controlset_\rank$ where $\rank$ is the minimum rank seen (his claim was correct).
In the second case, \theaplayer picks a rank $\rank$ and moves the play directly to some control state in $\controlset_\rank$.
This move has rank $\rank$ (as the minimum rank seen needs to contribute to the parity/safety condition).
In both cases, the topmost order-$(\cpdsord-1)$ stack does not change (as it would be the same in both plays).

To simulate a 
$\cpush{\cha}$
with
$\ap{\chorder}{\cha} = \cpdsord$
and a corresponding $\collapse$ we observe that the stack reached after the collapse is the same as that after a $\pop{\cpdsord}$ applied directly.
Thus, the simulation is similar.
To simulate the play up to the collapse, the current target set $\popguess$ is stored with the new stack character $\cha$.
Then \theeplayer wins if a move performs a $\collapse$ to a control state $\control \in \controlset_\rank$, where $\rank$ is the smallest rank seen since the order-$(\cpdsord-1)$ stack, that was topmost at the moment of the original $\cpush{\cha}$, was pushed.
To simulate the play after the collapse, we can simulate a $\pop{\cpdsord}$ as above.

\section{Counter Reduction}
\label{sec:counterreduction}

We reduce parity to safety games, generalizing Fridman and Zimmermann~\cite{FZ12} which extended Bernet\etal~\cite{BJW02}.
This reduction is not polynomial and we show in Section~\ref{sec:polynomial} how to achieve the desired complexity.
Correctness is \reflemma{lem:parityvssafety}.

We give the intuition here, with the detailed definition of $\ap{\counterred{\bound}}{\game}$ appearing in Appendix~\ref{sec:appendix-counterreduction}.
The reduction maintains a counter for each odd rank, which can take any value between $0$ and $\bound$.
We also detail the counters below as they are needed in Section~\ref{sec:polynomial}.

The insight of Bernet\etal is that, in a finite-state parity game of $\numof$ states, if \theaplayer can force the play to pass through some odd rank $\rank$ for $\numof+1$ times without visiting a state of lower rank in between, then some state $\control$ of rank $\rank$ is visited twice.
Since parity games permit positional winning strategies, \theaplayer can repeat the play from $\control$ ad infinitum.
Thus, the smallest infinitely occurring rank must be $\rank$, and \theaplayer wins the game.

Thus, \theeplayer plays a safety game: he must avoid visiting  an odd rank too many times without a smaller rank being seen.
In the safety game, counters
\[
    \countervec= \tup{\counter_1, \counter_3, \ldots, \counter_\maxrank}
\]
are added to the states, one for each odd rank.
When a rank $\rank$ is seen, then, if it is odd, $\counter_\rank$ is incremented.
Moreover, whether $\rank$ is odd or even, all counters
$\counter_{\rank'}$
for $\rank' > \rank$ are reset to $0$.

As the number of configurations is infinite, Bernet's insight does not immediately generalize to pushdown games.
However, Fridman and Zimmermann observed that, from Walukiewicz~\cite{W96}, a pushdown parity game can be reduced to a finite-state parity game (of exponential size) as described in the previous section.
This finite-state parity game can be further reduced to a safety game with the addition of counters.
Their contribution is then to transfer back the counters to the pushdown game, with the following reasoning.

Recall, a push move at
$\tup{\control, \sbrac{\cha \stackw}{1}}$
is translated into a branch from a corresponding state
$\tup{\control, \cha, \popguess}$
in the finite-state game.
There are several moves from
$\tup{\control, \cha, \popguess}$, some of them simulate the push, the remaining moves simulate the play after the corresponding pop.
When augmented with counters the states take the form
$\tup{\control, \cha, \popguess, \countervec}$.
We see that, when simulating the pop in the finite-state game, the counter values are the same as in the moment when the push is simulated.
That is, if we lift the counter construction to the pushdown game, after each pop move we need to reset the counters to their values at the corresponding push.
Thus we store the counter values on the stack.
For example, for a configuration
$\tup{\control, \sbrac{
    \tup{\cha, \countervec}
    \tup{\chb, \countervec'}
}{1}}$
where the current top of stack is $\cha$ and the current counter values are
$\countervec$, the counter values at the moment when $\cha$ was first pushed are stored on the stack as
$\countervec'$.

This reasoning generalizes to any order $\cpdsord$.
We store the counter values on the stack so that, when a $\pop{\opord}$ operation occurs, we can retrieve the counter values at the corresponding $\push{\opord}$, and similarly for $\collapse$.
Note also that, when reducing from order-$\cpdsord$ to order-\mbox{$(\cpdsord-1)$}, any branch corresponding to a play after a pop passes through a rank $\rank$ which is the smallest rank seen between the push and pop.
Thus, in the safety game, after each pop or collapse we need to update the counter values using $\rank$.
Hence we require a rank-aware CPDS.

Let $\maxrank$ be the maximum rank, and, for convenience, assume it is odd.
We maintain a vector of counters
$\countervec = \tup{\counter_1, \counter_3, \ldots, \counter_\maxrank}$,
one for each odd rank, stored in the stack alphabet as described above.
We update these counters with operations $\counterinc{\rank}$ that exist for all $\rank\in[0, m]$ (including the even ranks).
Operation $\counterinc{\rank}$ resets the counters $\counter_{\rank'}$ with $\rank' > \rank$ to zero.
If $\rank$ is odd, it moreover increments $\counter_\rank$.
If the bound is exceeded, an overflow occurs.
Formally, $\ap{\counterinc{\rank}}{\countervec}=\overflow$ if $\rank$ is odd and $\counter_\rank + 1 > \bound$.
Otherwise, $\ap{\counterinc{\rank}}{\countervec}=\countervec'$
where for each $\tilde \rank$
\[
    \counter'_{\tilde \rank} = \counter_{\tilde \rank}
    \text{ (if $\tilde \rank < \rank$)},
    \quad
    \counter'_{\tilde \rank} = \counter_{\rank} + 1
    \text{ (if $\tilde \rank = \rank$)},
    \text{ and}
    \quad
    \counter'_{\tilde \rank} = 0
    \text{ (if $\tilde \rank > \rank$).}
\]

\section{Equivalence Result}
\label{sec:equivalence}

We need equivalence of
$\gameoc=\ap{\counterred{\bound}}{\ap{\orderred}{\game}}$
and
$\gameco= \ap{\orderred}{\ap{\counterred{\bound}}{\game}}$
for Lemma~\ref{lem:commutes}.
The argument is that the two CPDS only differ in order of the components of their control states and stack characters.
A subtlety is that when $\counterred{\bound}$ is applied first, the contents of $\popguess$ are not control states of $\game$, but control states of $\gamec$.
However, the additional information in the control states after $\counterred{\bound}$ has to be consistent with $\popguess$, which means we can directly translate between guesses over states in the original CPDS, and those over states of the CPDS after the counter reduction.
The details are in Appendix~\ref{sec:equivalence-proof}.

\section{Polynomial Reduction}
\label{sec:polynomial}

For a game $\game$ over an order-$\cpdsord$ CPDS, the counters in the game
$\ap{\counterred{\ap{\bound}{\game}}}{\game}$
blow up $\game$ by an $\cpdsord$-fold exponential factor.
To avoid this we use the stack-like behaviour of the counters and a result due to Cachat and Walukiewicz~\cite{CW07}, showing how to encode large counter values into the stack of a CPDS with only polynomial overhead (in fact, collapse is not needed).
\subsection{Counter Encoding}

Cachat and Walukiewicz propose a binary encoding that is nested in the sense that a bit is augmented by its position, and the position is (recursively) encoded in the same way.
For example, number $5$ stored with $16$ bits is represented by
$(0, 1).(1, 0).(2, 1).(3, 0).(4, 0)\ldots(15, 0)$.
Since four bits are required to index $16$ bits, we encode position $4$ as
$(0, 0').(1, 0').(2, 1').(3, 0')$.
Finally, position $2$ of this encoding stored as $(0, 0'').(1, 1'')$.
The players compete to (dis)prove that the indexing is done properly.

Formally we introduce distinct alphabets to encode counters for all odd ranks $\rank$:
\begin{align*}
\counteralphabet{\rank}\ =\ \encodingalphabet{\rank} \cup \set{\bzero{\rank}, \bone{\rank}} \ .
\end{align*}
Here, $\encodingalphabet{\rank}$ is a polynomially-large set of characters for the indexing.
The set $\set{\bzero{\rank}, \bone{\rank}}$ are the bits to encode numbers.
Let $\counteralphabet{}$ be the union of all $\counteralphabet{\rank}$.

The values of the counters are stored on the order-$1$ stack, with the least-significant bit topmost.
The indices appear before each bit character.
E.g., value $16$ for counter $\rank$ stored with five bits yields a sequence from
$
    \encodingalphabet{\rank}^\ast\ .\ \bzero{\rank}\ .\ \encodingalphabet{\rank}^\ast\ .\ \bzero{\rank}\ .\ \encodingalphabet{\rank}^\ast\ .\ \bzero{\rank}\ .\
    \encodingalphabet{\rank}^\ast\ .\ \bzero{\rank}\ .\ \encodingalphabet{\rank}^\ast\ .\ \bone{\rank}
$.
Actually, the encoding will always use all bits, which means its length will be $(\cpdsord-1)$-fold exponential.

Cachat and Walukiewicz provide game constructions to assert properties of the counter encodings.
For this, play moves to a dedicated control state, from which \theeplayer wins iff the counters have the specified property.
In~\cite{CW07}, \theeplayer plays a reachability game from the dedicated state.
We need the dual, with inverted state ownership and a safety winning condition, where the target state of the (former) reachability game has rank $1$.
\theeplayer's goal will be to prove the encoding wrong (it violates a property) by means of safety,
\theaplayer tries to build up the counters correctly and, if asked, demonstrate correctness using reachability.

For all properties, the counter to be checked must appear directly on the top of the stack
(topmost on the topmost order-$1$ stack).
If any character outside $\counteralphabet{\rank}$ is found, \theaplayer loses.
When two counters are compared, the first counter must appear directly at the top of the stack, while the second may be separated from the first by any sequence of characters from outside $\counteralphabet{\rank}$
(these can be popped away).
The first character found from $\counteralphabet{\rank}$ begins the next encoding.
\theaplayer loses the game if none is found.
The required properties are listed below.
\begin{itemize}
\item
    Encoding Check ($\countercheck{\rank}$):
    For each rank $\rank$, we have a control state
    $\countercheck{\rank}$.
    \theaplayer can win the safety game from
    $\config{\countercheck{\rank}}{\stackw}$
    only if the topmost sequence of characters from
    $\counteralphabet{\rank}$
    is a correct encoding of a counter, in that all indices are present and correct.
\item
    Equals Check ($\equalscheck{\rank}$):
    For each $\rank$, we have a control state
    $\equalscheck{\rank}$, from which \theaplayer can win only if the topmost sequence of characters from
    $\counteralphabet{\rank}$
    is identical to the next topmost sequence of $\counteralphabet{\rank}$-characters.
    I.e., the two topmost $\rank$th counter encodings are equal.
\item
    Counter Increment:
    Cachat and Walukiewicz do not define increment but it can be done via the basic rules of binary addition.
    We force \theaplayer to increment the counter by first using $\pop{1}$ to remove characters from
    $\encodingalphabet{\rank} \cup \set{\bone{\rank}}$
    until $\bzero{\rank}$ is found.
    Then, \theaplayer must rewrite the $\bzero{\rank}$ to $\bone{\rank}$.
    \theaplayer then performs as many
    $\cpush{\cha}$
    operations as she wishes, where
    $\cha \in \encodingalphabet{\rank} \cup \set{\bzero{\rank}}$.
    Next, \theeplayer can accept this rewriting by continuing with the game, or challenge it by moving to $\countercheck{\rank}$.
    This ensures that \theaplayer has put enough $\bzero{\rank}$ characters on the stack (with correct indexing) to restore the number to its full length.
\end{itemize}

In this encoding one can only increment the topmost counter on the stack.
That is, to increment a counter, all counters above it must be erased.
Fortunately, $\counterinc{\rank}$ resets to zero all counters for ranks $\rank' > \rank$, meaning the counter updates follow a stack-like discipline.
This enables the encoding to work.
To store a character with counter values from the counter reduction
$\tup{\cha, \countervec}$ with $\countervec=\counter_1, \ldots, \counter_\maxrank$
we store the character $\cha$ on top and beneath we encode $\counter_\maxrank$, then $\counter_{\maxrank-2}$ and so on down to $\counter_1$.
\subsection{The Simulation}

The following definition is completed in the following sections.
Correctness is stated in \reflemma{lem:poly-equiv}.
The proof is formally detailed in Appendix~\ref{sec:poly-equiv-proof}.

\begin{definition}[$\polyred{\bound}$]
    Given a parity game
    $\game = \tup{\cpds, \owner, \wincond}$
    over the order-$\cpdsord$ CPDS
    $\cpds = \tup{\controls,
                  \alphabet, \chorder,
                  \cpdsrules,
                  \controlinit, \chainit,
                  \rankf}$
    and a bound $\bound$ $\cpdsord$-fold exponential in the size of the game, we define the safety game
    $\ap{\polyred{\bound}}{\game} =  \tup{\prcpds, \prowner, 2^{\omega}}$
    where
    $\prcpds =
    \tup{\prcontrols,
         \pralphabet, \prchorder,
         \prcpdsrules,
         \prcontrolinit, \prchainit,
         \prrankf}$.
    The missing components are defined below.
\end{definition}

We aim to simulate
$\ap{\counterred{\bound}}{\game}$ compactly.
This simulation is move-by-move, as follows.

A
$\cpush{\tup{\cha, \countervec}}$
of a character with counter values
$\tup{\cha, \countervec}$
with
$\countervec=\counter_1,\ldots, \counter_{\maxrank}$
(where the max-rank is $\maxrank$)
is simulated by first pushing a special character $\linkchar{\opord}$ to save the link (with $\cpush{\linkchar{\opord}}$).
Then, since the counter values are a copy of the preceding counter values on the stack, \theaplayer pushes an encoding for $\counter_1$ using $\counteralphabet{1}$ after which \theeplayer can accept the encoding, check that it is a proper encoding using $\countercheck{1}$, or check that it is a faithful copy of the preceding value of $\counter_1$ using $\equalscheck{1}$.
We do this for all odd ranks through to~$\maxrank$.
Then the only move is to push $\cha$ with
$\cpush{\cha}$.

Each $\push{\opord}$ and $\pop{\opord}$, with $\opord\in [2, n]$,
is simulated directly by the same operation.
For a $\pop{1}$ we (deterministically) remove all topmost characters (using $\pop{1}$) up to and including the first $\linkchar{\opord'}$ (for some $\opord'$).
We simulate $\collapse$ like $\pop{1}$, but we apply $\collapse$ to $\linkchar{\opord'}$.

A
$\rew{\tup{\cha, \countervec}}$
that does not change the counters can be simulated by rewriting the topmost character.
If $\counterinc{\rank}$ is applied, we force \theaplayer to play as follows.
If $\rank$ is even, \theaplayer removes the counters for $\rank' > \rank$.
She replaces them with zero values by pushing characters from
$\encodingalphabet{\rank'} \cup \set{\bzero{\rank'}}$.
After each counter is rewritten, \theeplayer can accept the encoding, or challenge it with $\countercheck{\rank'}$.
Finally, $\cha$ is pushed onto the stack.
If $\rank$ is odd, the counters for $\rank'$ are removed as before.
Then we do an increment as described above, with \theeplayer losing if the increment fails.
Note, it fails only if there is no $\bzero{\rank}$ in the encoding, which means the counter is at its maximum value and there is an overflow (indicating \theeplayer loses the parity game).
If it succeeds,
zero values for the counters $\rank' > \rank$ and $\cha$ are pushed to the stack as before.

\paragraph*{Control States and Alphabet}
We define the control states $\prcontrols$ with  $\prowner$ and $\prrankf$ as well as the alphabet.
First,
\begin{align*}
\prcontrols\ =\ \controls
        \cup
        \brac{\controls \times [0,\maxrank]}
        \cup
        \set{\eover, \aover}
        \cup
        \controlscw
        \ \cup \controlsop\ .
\end{align*}
where $\maxrank$ is the maximum rank.
The set $\controlscw$ is the control states of the Cachat-Walukiewicz games implementing $\countercheck{\rank}$ and $\equalscheck{\rank}$.
The size is polynomial in $\game$.
We have $\controlsop =$
\begin{align*}
    \begin{array}{c}
        \setcomp{
            \begin{array}{c}
                \controlincnoa{\rank}{\control},
                \controlcopy{\rank}{\cha}{\control},
                \controlzero{\rank}{\cha}{\control},
                \controlpop{\genplayer}{\rank}{\control}, \\
                \controlinc{\rank}{\cha}{\control},
                \controlcopytest{\rank}{\cha}{\control},
                \controlzerotest{\rank}{\cha}{\control},
                \controlcollapse{\genplayer}{\rank}{\control}
            \end{array}
        }{
            \begin{array}{c}
                \rank \in [0, \maxrank] \land \cha \in \alphabet
                \ \land \\
                \control \in \controls\ \land \genplayer \in \set{\aplayer,\eplayer}
            \end{array}
        }
    \end{array}
\end{align*}
to control the simulation of the operations as sketched above.
We describe the states below.

The states in $\controlscw$ have the same rank and owner as in the Cachat-Walukiewicz games
(more precisely the dual, see above).
All other states have rank $2$ except $\eover$ which has rank $1$.
It (resp.\ $\aover$) is the losing sink for \theeplayer (resp.\ \theaplayer).
The states in
$\controls \cup \brac{\controls \times [0, \maxrank]}\cup\set{\eover}$
are used as in
$\ap{\counterred{\bound}}{\game}$ to directly simulate $\game$.
The owners are as in $\game$.

A state
$\controlincnoa{\rank}{\control}$
begins an application of $\counterinc{\rank}$.
The top-of-stack character is saved by moving to
$\controlinc{\rank}{\cha}{\control}$.
The owner of these states does not matter, we give them to \theaplayer.
In $\controlinc{\rank}{\cha}{\control}$, the stack is popped down to the counter for $\rank$.
If $\rank$ is odd, the least significant zero is set to one.
Then, control moves to $\controlzero{\rank}{\cha}{\control}$.
In $\controlzero{\rank}{\cha}{\control}$, zero counters for ranks $\rank$ and above are pushed to the stack, followed by a push of $\cha$ and a return to control state $\control$.
The state is owned by \theaplayer.
The state
$\controlzerotest{\rank}{\cha}{\control}$
is used by \theeplayer to accept or challenge that the encoding has been re-established completely.
It is owned by \theeplayer.

The controls
$\controlcopy{\rank}{\cha}{\control}$
copy the counters for ranks $\rank$ and above (the current values) and push the copies to the stack, followed by a push of $\cha$ and a return to control state $\control$.
The state is owned by \theaplayer.
After this phase, the play moves to
$\controlcopytest{\rank}{\cha}{\control}$
where \theeplayer can accept or test whether the copy has been done correctly.
This state is owned by \theeplayer.

The controls
$\controlpop{\genplayer}{\rank}{\control}$
and
$\controlcollapse{\genplayer}{\rank}{\control}$
where
$\genplayer \in \set{\aplayer,\eplayer}$
are used to execute a $\pop{1}$ or $\collapse$.
For the latter,  we pop to the next $\linkchar{\opord}$ character, perform the collapse and record that the $\rank$th counter needs to be incremented.
In case the collapse is not possible (because it would empty the stack) play may also move to a sink state that is losing for the player $\genplayer$ who instigated the collapse.
The case of $\pop{1}$ pops $\linkchar{\opord}$.
The owner in each case is $\genplayer$ as they will avoid moving to their (losing) sink state if the $\pop{1}$ or $\collapse$ is possible.

The alphabet and initial control state and stack character are
\[
    \alphabet' = \alphabet \cup
                    \counteralphabet{}
                    \cup
                 \linkalphabet
    \quad
    \text{and}
    \quad
    \prcontrolinit = \controlzero{1}{\chainit}{\controlinit}
    \quad
    \text{and}
    \quad
    \prchainit = \linkchar{\opord}\text{ where }\ap{\chorder}{\chainit} = \opord\ .
\]
The alphabet is extended by the characters required for the counter and link encodings.
Recall that $\counteralphabet{}$ is the union of the counter alphabets, which are of polynomial size.
We use $\linkalphabet=\set{\linkchar{1}, \ldots, \linkchar{\maxrank}}$ for the link characters.
We assign
$\ap{\prchorder}{\linkchar{\opord}} = \opord$
and
$\ap{\prchorder}{\cha} = 1$
for all other $\cha$.

The task of the initial state and initial stack character is to establish the encoding of $(\chainit, (0,\ldots, 0))$ in $\ap{\counterred{\bound}}{\game}$ and then move to the initial state of $\game$.
With the above description, $\controlzero{1}{\chainit}{\controlinit}$ will establish zeros in all counters (from $1$ to $\maxrank$), push the initial character $\chainit$ of the given game, and move to state $\controlinit$.
The initial character $\linkchar{\opord}$ is the correct bottom element for the encoding of $(\chainit, (0,\ldots, 0))$.
\paragraph*{Rules}

The rules of $\prcpds$ follow $\cpds$ and maintain the counters.
$\prcpdsrules$ contains (only) the following rules.
First, we have $\rulescw$ which are the (dual of the) rules of Cachat and Walukiewicz implementing $\countercheck{\rank}$ and $\equalscheck{\rank}$.
The rules simulating the operations appear below.
Note, pop and collapse use rank-awareness.
We give the increment and copy rules after the basic operations.
\begin{itemize}
\item
    Order-$\opord$ push:
    $\tup{\control,
          \cha,
          \push{\opord},
          \controlincnoa{\ap{\rankf}{\control'}}{\control'}}$
    when
        $\tup{\control, \cha, \push{\opord}, \control'} \in \cpdsrules$.
\item
    Character push:
    $\tup{\control,
          \cha,
          \cpush{\linkchar{\opord}},
          \controlcopy{1}{\chb}{\control'}}$
    when
        $\tup{\control, \cha, \cpush{\chb}, \control'} \in \cpdsrules$ and
        $\ap{\chorder}{\chb} = \opord$.

\item
    Rewrite:
    $\tup{\control, \cha, \rew{\chb}, \controlincnoa{\ap{\rankf}{\control'}}{\control'}}$
    when
        $\tup{\control, \cha, \rew{\chb}, \control'} \in \cpdsrules$.

\item
    Pop ($> 1$):
    $\tup{\control, \cha, \pop{\opord}, \controlincnoa{\rank}{\control'}}$
    when
        $\tup{\control, \cha, \pop{\opord}, \control'} \in \cpdsrules$ and
        $\rank\ =\ \ap{\min}{\ap{\rankf}{\control'}, \ap{\colrankfun}{\opord}}$.

\item
    Pop ($= 1$) and Collapse:
    $\tup{\control, \cha, \pop{1}, \tup{\genop, \genplayer, \rank, \control'}}$
    when
        $\tup{\control, \cha, \genop, \control'} \in \cpdsrules$
        with
        $\ap{\owner}{\control} = \genplayer$,
        operation $\genop$ being $\pop{1}$ or $\collapse$, and
        $\rank =\ \ap{\min}{
            \ap{\rankf}{\control'},
            \rank'
        }$.
        Here, if $\genop = \pop{1}$ then
        $\rank' = \ap{\colrankfun}{1}$.
        Otherwise, if $\genop = \collapse$ then
        $\rank' = \ap{\colrankfun}{\link}$.

    Then, we have all rules
    $\tup{\tup{\genop, \genplayer, \rank, \control'},
          \cha',
          \pop{1},
          \tup{\genop, \genplayer, \rank, \control'}}$
    for
    $\cha' \in \counteralphabet{}$.
    We perform the operation with
    $\tup{\tup{\genop, \genplayer, \rank, \control'},
          \linkchar{\opord'},
          \genop,
          \controlincnoa{\rank}{\control'}}$.
    To allow for the case where the pop or collapse cannot be performed (because the stack would empty), we also have the rules
    $\tup{\tup{\genop, \aplayer, \rank, \control'},
          \linkchar{\opord'},
          \noop,
          \aover}$
    and
    $\tup{\tup{\genop, \eplayer, \rank, \control'},
         \linkchar{\opord'},
         \noop,
         \eover}$.

  \item
    Sink states:
    $\tup{\aover, \cha, \noop, \aover}$
    and
    $\tup{\eover, \cha, \noop, \eover}$.
\end{itemize}

To copy counters, for each odd $\rank$ and
$\chb \in \counteralphabet{\rank}$
we have
$\tup{
    \controlcopy{\rank}{\cha}{\control},
    \dontcare,
    \cpush{\chb},
    \controlcopy{\rank}{\cha}{\control}
}$.
We use $\dontcare$ to indicate that the transition exists for all stack symbols.
When a counter has been pushed, like in the case of pushing zeros, \theaplayer hands over the control to \theeplayer to check the result:
$\tup{
    \controlcopy{\rank}{\cha}{\control},
    \dontcare,
    \noop,
    \controlcopytest{\rank}{\cha}{\control}
}$.
\theeplayer can challenge the copied counter or accept it was copied correctly.
To challenge, we use
$\tup{
    \controlcopytest{\rank}{\cha}{\control},
    \dontcare,
    \noop,
    \equalscheck{\rank}
}$.
To accept, the behavior depends on $\rank$.
If $\rank < \maxrank$, we move to copying the next counter
$
\tup{
    \controlcopytest{\rank}{\cha}{\control},
    \dontcare,
    \noop,
    \controlcopy{\rank+2}{\cha}{\control}
}$.
When $\rank = \maxrank$, we finish copying and move to incrementing with rules of the form
$\tup{
    \controlcopytest{\rank}{\cha}{\control},
    \dontcare,
    \noop,
    \controlinc{\ap{\rankf}{\control}}{\cha}{\control}
}$.

To increment a counter, we first pop and store the topmost stack character with the rule
$\tup{\controlincnoa{\rank}{\control},
          \cha,
          \pop{1},
          \controlinc{\rank}{\cha}{\control}}$.
\theaplayer then removes all counters for ranks higher than the given $\rank$ with the following rules, where $\chb \in\counteralphabet{\rank'}$ with $\rank' > \rank$:
$\tup{\controlinc{\rank}{\cha}{\control},
             \chb,
             \pop{1},
             \controlinc{\rank}{\cha}{\control}}$.

When $\rank$ is even we add back $0$ counters once enough have been removed using
(with
 $\chb \in \counteralphabet{\rank-1}$
 if $\rank > 1$ else
 $\chb \in \linkalphabet$)
the rules
$\tup{\controlinc{\rank}{\cha}{\control},
             \chb,
             \noop,
             \controlzero{\rank+1}{\cha}{\control}}$.
If $\rank$ is odd, we start incrementing the $\rank$th counter with
$\tup{\controlinc{\rank}{\cha}{\control},
      \chb,
      \pop{1},
      \controlinc{\rank}{\cha}{\control}}$
for all
$\chb \in \encodingalphabet{\rank} \cup \set{\bone{\rank}}$.

When $\bzero{\rank}$ is found, we use
$\tup{\controlinc{\rank}{\cha}{\control},
             \bzero{\rank},
             \rew{\bone{\rank}},
             \controlzero{\rank}{\cha}{\control}}$.
If no zero bit is found, we have an overflow and move
to the sink state with $\tup{\controlinc{\rank}{\cha}{\control},
          \chb,
          \noop,
          \eover}$
    for
    $\chb \in \counteralphabet{\rank-2}$
    if $\rank > 2$ and
    $\chb \in \linkalphabet$
    otherwise.
With
$\tup{
        \controlzero{\rank}{\cha}{\control},
        \dontcare,
        \cpush{\chb},
        \controlzero{\rank}{\cha}{\control}
    }$
for
$\chb\in \encodingalphabet{\rank} \cup \set{\bzero{\rank}}$
we add back zeros to the incremented counter and reset all erased counters.
To finish the phase that adds zeros for the $\rank$th counter, \theaplayer hands over the control to \theeplayer,
$\tup{
        \controlzero{\rank}{\cha}{\control},
        \dontcare,
        \noop,
        \controlzerotest{\rank}{\cha}{\control}
    }$.

\theeplayer can now check if all bits of the counter are present or accept the result.
To challenge the encoding, he uses
$\tup{
    \controlzerotest{\rank}{\cha}{\control},
    \dontcare,
    \noop,
    \countercheck{\rank}
}$.
When accepting it, if $\rank < \maxrank$, more counters need to be reset.
We move to the next using
$\tup{
        \controlzerotest{\rank}{\cha}{\control},
        \dontcare,
        \noop,
        \controlzero{\rank+2}{\cha}{\control}
    }$.
If $\rank = \maxrank$, there are no more counters to handle and with the rule
$\tup{
    \controlzerotest{\rank}{\cha}{\control},
    \dontcare,
    \cpush{\cha},
    \control
}$
\theeplayer re-establishes the control state and stack character.

\section{Conclusion}
\label{sec:conclusion}
We gave a polynomial-time reduction from parity games played over order-$\cpdsord$ CPDS to safety games over order-$\cpdsord$ CPDS. 
Such a reduction has been an open problem~\cite{FZ12} (related are also \cite{BJW02,BD08}). 
It builds counters into the stack to count occurrences of odd ranks at the current stack level (without seeing a smaller rank). 
If this number grows large then \theeplayer would lose the parity game (if play continued). 
To obtain a polynomial reduction we use the insight that the counters follow a stack discipline.
For correctness, we use a commutativity argument for the rank counter and order reductions. 
As a theoretical interest, the result explains the matching complexities of parity and safety games over CPDS.
From a practical standpoint, the reduction may inspire the use of advanced safety checking tools for, and the transfer of technology from safety to, the empirically harder problem of parity game analysis.


\bibliography{57hmmz}

\newpage
\appendix

\section{Order-$1$ Rank-Aware}
\label{sec:order-1-rank-aware}

\reflemma{lem:rank-aware} was shown by Hague\etal~\cite{HMOS08} in a slightly restricted form.
In particular, the provided rank function $\colrankfun$ does not provide information on the order-$1$ component.
That is
\[
    \colrankfun : [2,\cpdsord] \cup \set{\link} \rightarrow
                  [0,\maxrank]
\]
where $\maxrank$ is the maximum rank.
Let us call this \emph{higher-rank-aware}.
Let $\higherranks$ be the set of all such functions.
We show that this can easily be extended to rank-awareness as defined in \refdefinition{def:rank-aware} where
\[
    \colrankfun : [1,\cpdsord] \cup \set{\link} \rightarrow
                  \set{0,\ldots,\maxrank}
\]
as required.
Let $\allranks$ be the set of all such functions.
Moreover, given
$\colrankfun \in \higherranks$,
let
$\extrkf{\colrankfun}{\rank}$
be
\[
    \ap{\extrkf{\colrankfun}{\rank}}{\opord} =
    \begin{cases}
        \rank & \opord = 1 \\
        \ap{\colrankfun}{\rank} & \text{otherwise.}
    \end{cases}
\]

Assume we are given a parity game over a CPDS that is higher-rank-aware.
We can transform this into an equivalent rank-aware CPDS by simply storing in each character the minimum rank seen since the character was pushed.
The definition is given below.

\begin{definition}
    Let
    $\game = \tup{\cpds, \controlinit, \chainit, \owner, \rankf, \wincond}$
    be a parity game over a CPDS
    \[
        \cpds = \tup{\controls, \alphabet \times \higherranks, \cpdsrules}
    \]
    such that $\cpds$ is higher-rank-aware.
    We define
    \[
        \game' =
            \tup{\cpds', \controlinit', \chainit', \owner', \rankf', \wincond}
    \]
    where
    \[
        \cpds' = \tup{\controls', \alphabet', \cpdsrules'}
    \]
    and
    $\controls' = \controls \cup \controls \times \set{0,\ldots,\maxrank}$,
    $\alphabet' = \alphabet \times \allranks$,
    $\cpdsrules'$ is the smallest set containing
    \begin{itemize}
    \item
        $\tup{\control,
              \tup{\cha, \extrkf{\colrankfun}{\rank}},
              \cpush{\tup{\chb,
                          \extrkf{\colrankfun'}
                                 {\ap{\rankf}{\control'}}}},
              \control'}$
        whenever we have the rule
        $\tup{\control,
              \tup{\cha,\colrankfun},
              \cpush{\tup{\chb,\colrankfun'}},
              \control'}
         \in
         \cpdsrules$
        and
        $0 \leq \rank \leq \maxrank$,

    \item
        $\tup{\control,
              \tup{\cha, \extrkf{\colrankfun}{\rank}},
              \rew{\tup{\chb,
                        \extrkf{\colrankfun'}{\rank'}}},
              \control'}$
        whenever we have a rule of the form
        $\tup{\control,
              \tup{\cha, \colrankfun},
              \rew{\tup{\chb, \colrankfun'}},
              \control'} \in \cpdsrules$
        and
        $0 \leq \rank \leq \maxrank$,
        where
        $\rank' = \ap{\min}{\rank, \ap{\rankf}{\control'}}$,

    \item
        $\tup{\control,
              \tup{\cha, \extrkf{\colrankfun}{\rank}},
              \genop,
              \tup{\control', \rank'}}$
        for each
        $\tup{\control,
              \tup{\cha, \colrankfun},
              \genop,
              \control'} \in \cpdsrules$,
        $\rank' = \ap{\min}{\rank, \ap{\rankf}{\control'}}$
        and
        $0 \leq \rank \leq \maxrank$
        where $\genop$ is
        $\push{\opord}$, $\pop{\opord}$, or $\collapse$
        for some $\opord$, and

    \item
        $\tup{\tup{\control, \rank},
              \tup{\cha, \extrkf{\colrankfun}{\rank'}},
              \rew{\tup{\cha, \extrkf{\colrankfun}{\rank''}}},
              \control}$
        for each
        $\control \in \controls$,
        $\cha \in \alphabet$,
        $\colrankfun \in \higherranks$,
        $0 \leq \rank, \rank' \leq \maxrank$, and
        $\rank' = \ap{\min}{\rank, \rank'}$,
    \end{itemize}
    the ownership function is defined
    \[
        \ap{\owner'}{\control} =
        \begin{cases}
            \ap{\owner}{\control'} & \control = \tup{\control', \rank} \\
            \ap{\owner}{\control} & \text{otherwise.}
        \end{cases}
    \]
    We define the initial control state and stack character
    \[
        \begin{array}{rcl}
            \controlinit' &=& \controlinit \\
            \chainit' &=& \tup{\cha,
                               \extrkf{\colrankfun}
                                      {\ap{\rankf}{\controlinit}}}
        \end{array}
    \]
    where
    $\chainit = \tup{\cha, \colrankfun}$.
\end{definition}

One can observe the equivalence between $\game$ and $\game'$ from the initial configuration
$\tup{\controlinit, \sbrac{\cdots\sbrac{\chainit}{1}\cdots}{\cpdsord}}$
of $\game$ and the initial configuration
$\tup{\controlinit',
      \sbrac{\cdots\sbrac{
          \tup{\chainit, \ap{\rankf}{\controlinit}}
      }{1}\cdots}{\cpdsord}}$
of $\game'$ by noting
\begin{itemize}
\item
    there is a move
    $\tup{\control,
          \tup{\cha,\colrankfun},
          \cpush{\tup{\chb,\colrankfun'}},
          \control'}
     \in
     \cpdsrules$
    in $\game$ iff there is a corresponding move
    $\tup{\control,
          \tup{\cha, \extrkf{\colrankfun}{\rank}},
          \cpush{\tup{\chb,
                      \extrkf{\colrankfun'}
                             {\ap{\rankf}{\control'}}}},
          \control'}$
    in $\game'$,

\item
    there is a move
    $\tup{\control,
          \tup{\cha, \colrankfun},
          \rew{\tup{\chb, \colrankfun'}},
          \control'} \in \cpdsrules$
    in $\game$ iff there is a corresponding move
    $\tup{\control,
          \tup{\cha, \extrkf{\colrankfun}{\rank}},
          \rew{\tup{\chb,
                    \extrkf{\colrankfun'}{\rank'}}},
          \control'}$
    in $\game'$ where
    $\rank' = \ap{\min}{\rank, \ap{\rankf}{\control'}}$,

\item
    a move
    $\tup{\control, \tup{\cha, \colrankfun}, \genop, \control'}$
    can be played in $\game$ iff a pair of moves
    $\tup{\control,
          \tup{\cha, \extrkf{\colrankfun}{\rank_1}},
          \genop,
          \tup{\control', \rank_2}}$
    and
    \[
        \tup{\tup{\control', \rank_2},
             \tup{\chb, \extrkf{\colrankfun'}{\rank_3}},
             \rew{\tup{\chb, \extrkf{\colrankfun'}{\rank_4}}},
             \control'}
    \]
    for all other $\genop$ where
    $\rank_4 = \ap{\min}{\rank_2, \rank_3}$
    can be played in $\game'$.
    Note, the second move may read $\chb$ rather than $\cha$ since $\genop$ may have been a pop or collapse operation.
\end{itemize}
Thus, a play of $\game$ can be directly mapped to a play in $\game'$ and vice-versa.
In both directions, the sequence of ranks generated is the same, except each rank may appear twice (contiguously).
Thus, the winner remains the same.

Moreover, one can observe that $\cpds'$ is collapse-rank aware\footnote{%
    Strictly speaking, we have a one-move delay after some operations.
    This, however, is benign and a feature of the original rank-awareness construction.
}.
For all orders except $1$ this follows from $\colrankfun$.
For order-$1$, observe that when a character $\cha$ is pushed it is tagged with the rank of the current control state
($\tup{\cha, \ap{\rankf}{\control'}}$).
The invariant maintained is that when
$\tup{\cha, \rank}$
is at the top of the stack, then $\rank$ is the smallest rank seen since $\cha$ was pushed.
One can observe that each move of $\game'$ maintains this invariant, either directly
(in the case of rewrite moves)
or via an intermediate control state
$\tup{\control, \rank'}$
which indicates the current top character needs to be updated to account for a (smallest) rank of $\rank'$ having been seen since the character was last updated.

\begin{proposition}
    Elvis wins $\game$ iff he wins $\game'$ and $\game'$ is collapse-rank aware.
\end{proposition}

\section{Precise Calculation of Bound} \label{sec:bound-calc}

We show how to obtain the required bound precisely.
Let $\maxrank$ be the maximum rank and $\numranks = \maxrank + 1$ and assume $\numranks \geq 2$.
At order-$1$, from Fridman and Zimmerman the bound must be at least
$\sizeof{\controls}
 \cdot
 \sizeof{\alphabet}
 \cdot
 2^{\numranks \cdot \sizeof{\controls}}
 \cdot
 \numranks$.
Recall $\sizeof{\cpds} = \sizeof{\controls} + \sizeof{\Sigma}$.
We will use $\sizeof{\cpds}$ below.
For the generalization of this bound, we need to consider the cost of the rank-awareness transformation.
This is bound by some polynomial $f(x)$, which can be bound by
$\rankawareconst \cdot x^\rankawarepower$
for some constants $\rankawareconst$ and $\rankawarepower$.
For technical reasons, assume
$\rankawarepower \geq 9$.
In practice, one may produce a more fine-grained analysis to keep the bounds as low as possible.

Given a parity game
$\game = \tup{\cpds, \owner, \wincond}$
over an order-$\cpdsord$ CPDS
$\cpds = \tup{\controls,
              \alphabet, \chorder,
              \cpdsrules,
              \controlinit, \chainit,
              \rankf}$
we set a bound of
\[
    \ap{\nexp{\cpdsord}}{
        \rankawareconst \cdot
        16^\rankawarepower \cdot
        \numranks^{2 \cdot \rankawarepower} \cdot
        \sizeof{\cpds}^{2 \cdot \rankawarepower}
    }
\]
which exceeds the required order-$1$ bound.
The number of control states in
$\gameo = \ap{\orderred}{\game}$
will be seen to be bounded by
$4 \cdot
 \sizeof{\cpds}^2 \cdot
 \numranks^2 \cdot
 2^{2 \cdot \sizeof{\cpds} \cdot \numranks}$.
The size of the alphabet will be seen to be bound by
$\sizeof{\cpds} + \sizeof{\cpds} \cdot 2^{\sizeof{\cpds} \cdot \numranks}$.
Thus, the size of the CPDS will be bound by
$6 \cdot
 \sizeof{\cpds}^2 \cdot
 \numranks^2 \cdot
 2^{2 \cdot \sizeof{\cpds} \cdot \numranks}$.
This becomes
$\rankawareconst \cdot
 6^\rankawarepower \cdot
 \sizeof{\cpds}^{2 \cdot \rankawarepower} \cdot
 \numranks^{2 \cdot \rankawarepower} \cdot
 2^{2 \cdot \sizeof{\cpds} \cdot \numranks \cdot \rankawarepower}$
after re-establishing rank-awareness.

Then, by induction, we know that \theeplayer will win the parity game $\gameo$ iff he wins $\gameoc$ for any
\[
    \bound
    \geq
    \ap{\nexp{\cpdsord-1}}{
        \begin{array}{c}
            \rankawareconst \cdot
            16^\rankawarepower \cdot
            \numranks^{2 \cdot \rankawarepower}\ \cdot
            \rankawareconst^{2 \cdot \rankawarepower} \cdot
            6^{2 \cdot \rankawarepower^2} \cdot
            \sizeof{\cpds}^{4 \cdot \rankawarepower^2} \cdot
            \numranks^{4 \cdot \rankawarepower^2} \cdot
            2^{4 \cdot \sizeof{\cpds} \cdot \numranks \cdot \rankawarepower^2}
        \end{array}
    } \ .
\]
We argue that each of the $8$ terms in the product is bound by
$2^{\rankawareconst \cdot
    2^{\rankawarepower} \cdot
    \numranks^{2 \cdot \rankawarepower} \cdot
    \sizeof{\cpds}^{2 \cdot \rankawarepower}}$
as follows:
\begin{enumerate}
\item
    $\rankawareconst \leq 2^\rankawareconst$,
\item
    $16^\rankawarepower \leq 2^{2^\rankawarepower}$
    when
    $\rankawarepower \geq 4$,
\item
    $\numranks^{2 \cdot \rankawarepower} \leq 2^{\numranks^{2 \cdot \rankawarepower}}$,
\item
    $\rankawareconst^{2 \cdot \rankawarepower}
     \leq
     2^{\rankawareconst \cdot 2^\rankawarepower}$
    since
    $\rankawareconst^{2 \cdot \rankawarepower}
     \leq
     2^{\rankawareconst \cdot 2 \cdot \rankawarepower}$
    and
    $\rankawareconst \cdot 2 \cdot \rankawarepower
     \leq
     \rankawareconst \cdot 2^\rankawarepower$,
\item
    $6^{2 \cdot \rankawarepower^2} \leq 2^{2^\rankawarepower}$
    when
    $\rankawarepower \geq 9$,
\item
    $\sizeof{\cpds}^{4 \cdot \rankawarepower^2}
     \leq
     2^{\sizeof{\cpds}^{2 \cdot \rankawarepower}}$
    because
    $\sizeof{\cpds}^{4 \cdot \rankawarepower^2}
     \leq
     2^{\sizeof{\cpds} \cdot 4 \cdot \rankawarepower^2}$
    and, comparing the exponents, we show
    $\sizeof{\cpds} \cdot 4 \cdot \rankawarepower^2
     \leq
     \sizeof{\cpds}^{2 \cdot \rankawarepower}$
    by dividing by $\sizeof{\cpds}$ to get
    $4 \cdot \rankawarepower^2
     \leq
     \sizeof{\cpds}^{2 \cdot \rankawarepower - 1}$
    which holds since $\sizeof{\cpds}$ is at least $2$
    (one control state, one character)
    and
    $4 \cdot \rankawarepower^2
     \leq
     2^{2 \cdot \rankawarepower - 1}$
    holds for all $\rankawarepower \geq 4$,
\item
    $\numranks^{4 \cdot \rankawarepower^2}
     \leq
     2^{\numranks^{2 \cdot \rankawarepower}}$
    since
    $\numranks^{4 \cdot \rankawarepower^2}
     \leq
     2^{\numranks \cdot 4 \cdot \rankawarepower^2}$
    and comparing the exponents we show
    $\numranks \cdot 4 \cdot \rankawarepower^2
     \leq
     \numranks^{2 \cdot \rankawarepower}$
    by dividing by $\numranks$ and observing we have
    $4 \cdot \rankawarepower^2 \leq \numranks^{2 \cdot \rankawarepower - 1}$
    if both $\numranks \geq 2$ and $\rankawarepower \geq 4$.
\item
    $2^{4 \cdot \sizeof{\cpds} \cdot \numranks \cdot \rankawarepower^2}
     \leq
     2^{2^\rankawarepower \cdot
        \numranks^{2 \cdot \rankawarepower} \cdot
        \sizeof{\cpds}^{2 \cdot \rankawarepower}}$
    because immediately
    $\sizeof{\cpds} \leq \sizeof{\cpds}^{2 \cdot \rankawarepower}$
    and
    $\numranks \leq \numranks^{2 \cdot \rankawarepower}$
    and the remaining exponents have
    $4 \cdot \rankawarepower^2 \leq 2^\rankawarepower$
    for all
    $\rankawarepower \geq 8$.
\end{enumerate}

For $\gameoc$, the assigned bound
\[
    \bound =
        \ap{\nexp{\cpdsord}}{
            \rankawareconst \cdot
            16^\rankawarepower \cdot
            \numranks^{2 \cdot \rankawarepower} \cdot
            \sizeof{\cpds}^{2 \cdot \rankawarepower}
        }
\]
is large enough because it exceeds
\[
    \ap{\nexp{\cpdsord-1}}{
        \brac{
            2^{\rankawareconst \cdot
               2^\rankawarepower \cdot
               \numranks^{2 \cdot \rankawarepower} \cdot
               \sizeof{\cpds}^{2 \cdot \rankawarepower}}
        }^8
    } \ .
\]
Thus, with $\bound$, the diagram in Section~\ref{sec:outline} commutes.

\section{Details of Order Reduction}
\label{sec:appendix-orderreduction}

We complete the following definition in the sequel.

\begin{definition}{$\orderred$}
    Given an order-$\cpdsord$ CPDS
    $\cpds = \tup{\controls,
                  \alphabet, \chorder,
                  \cpdsrules,
                  \controlinit, \chainit,
                  \rankf}$
    and game
    $\game = \tup{\cpds, \owner, \wincond}$
    we define
    $\ap{\orderred}{\game} = \tup{\orcpds, \orowner, \wincond}$
    where
    $\orcpds = \tup{\orcontrols,
                    \oralphabet, \orchorder,
                    \orcpdsrules,
                    \orcontrolinit, \chainit,
                    \orrankf}$.
\end{definition}

Let $\maxrank$ be the maximum rank in the game.

\subsection{Control States and Alphabet}

Let $\popguess, \popguess'$ range over $\powerset{\controls}^{\maxrank+1}$.
The set of control states $\orcontrols$ contains all states of the following five forms.
We also define the owner ($\orowner$) and rank ($\orrankf$) of each control state.
When simulating commands different from $\push{\cpdsord}$ and $\cpush{\cha}$ with
$\ap{\chorder}{\cha} = \cpdsord$,
the control states are $\orcontrol=\tup{\control, \popguess}$.
Here, $\control$ is the current control state and $\popguess$ is the current claim about where the play will $\pop{\cpdsord}$ to.
    The owner is $\ap{\owner}{\control}$.
    The rank is $\ap{\rankf}{\control}$.
The remaining control states are as follows, where we sometimes access the component $\control$ of $\orcontrol$:
\begin{itemize}
\item
    $\tup{\orcontrol, \control'}$
    is used to begin the simulation of a $\push{\cpdsord}$ to control state $\control'$.
    It is a tuple $\orcontrol$ as before plus $\control'$, which is the target control state of the push.
    These control states are owned by \theeplayer.
    The rank is $\ap{\rankf}{\control}$.

\item
    $\tup{\orcontrol, \control', \popguess'}$
    is used to continue the simulation of $\push{\cpdsord}$.
    The extra component $\popguess'$ is the guess of \theeplayer about where the play will pop to.
    These control states are owned by \theaplayer.
    The rank is $\ap{\rankf}{\control}$.

\item
    $\tup{\orcontrol, \rank}$
    (or $\tup{\control, \popguess, \rank}$)
    is used after skipping directly to the pop location of a push.
    Here, $\rank$ is the smallest rank that would have been seen if the push to pop was performed in full.
    The rank is $\rank$.
    The owner does not matter.
    For completeness we give the states to \theeplayer.

\item
    $\ewin$, $\elose$ are sink states indicating that \theeplayer won or lost the game, respectively.
    The ownership of these states does not affect the game, we give them to \theeplayer.
\end{itemize}
The initial control state is
\begin{align*}
\orcontrolinit=
        \tup{\controlinit, \undefpopguess}
\end{align*}
where $\undefpopguess$ is a special value indicating that a $\pop{\cpdsord}$ operation cannot occur (as this would empty the stack).
The stack alphabet $\oralphabet$ is
$\alphabet \cup \brac{
    \alphabet \times \brac{
        \powerset{\controls}^{\maxrank+1} \cup
        \set{\undefpopguess}
    }
}$.
A letter
$\tup{\cha, \popguess}$
contains the obligation to reach $\popguess$ when $\collapse$ is applied to the letter
(or indicates collapse is not possible if $\popguess = \undefpopguess$).
Finally, we define
$\ap{\orchorder}{\cha'} = 1$
when
$\cha' = \tup{\cha, \popguess}$
and
$\ap{\orchorder}{\cha'} = \ap{\chorder}{\cha'}$
otherwise.

\subsection{Rules}
We first describe the rules directly simulating $\cpds$, then we give the rules required to simulate order-$\cpdsord$ rules.
Starting from a rule $\tup{\control, \cha, \genop, \control'} \in \cpdsrules$, we define the rules of $\orcpdsrules$ by a case distinction along the operations $\genop$.
To shorten the definition, we let $\orcontrol = \tup{\control, \popguess}$ and $\orcontrol' = \tup{\control', \popguess}$.
Here, $\control$ and $\control'$ stem from the simulated rule.
The component $\popguess$ is implicitly universally quantified, which means there are rules in $\orcpdsrules$ for all choices of $\popguess$.
We use this implicit universal quantification for all components that are not defined explicitly.
Unless otherwise stated, all definitions are assumed to require the constraint $\cha' = \cha$ or $\cha' = \tup{\cha, \popguess^{\natural}}$ (where $\popguess^{\natural}$ will be universally quantified over
$\powerset{\controls}^{\maxrank+1} \cup
 \set{\undefpopguess}$).
The direct simulation of rules is by two definitions.
\begin{itemize}
\item
    Operations: $\tup{\orcontrol,
          \cha', \genop,
          \orcontrol'}$
    when
        $\genop$ is not order-$\cpdsord$ and not $\rew{\chb}$
        and if $\genop = \collapse$ then $\cha'$ is not of the form
        $\tup{\cha, \popguess^{\natural}}$.
        \vspace{0.1cm}

\item
    Rewrite: $\tup{\orcontrol,
          \cha', \rew{\chb'},
          \orcontrol'}$
    when $\genop$ is  $\rew{\chb}$
        and either
            $\cha' = \cha$ and $\chb' = \chb$, or
            $\cha' = \tup{\cha, \popguess^{\natural}}$
            and
            {$\chb' = \tup{\chb, \popguess^{\natural}}$}\ .
\item
    Push: the following rules simulate $\genop=\push{\cpdsord}$ :
    \begin{itemize}
    \item
        Start: $\tup{\orcontrol,
              \cha', \noop,
              \tup{\orcontrol, \control'}}$. Note that $\control'$ is from the rule in $\cpdsrules$.\vspace{0.1cm}

    \item
        Claim:
        $\tup{\tup{\orcontrol, \control'},
              \cha', \noop,
              \tup{\orcontrol, \control', \popguess'}}$.
        Here, the obligation $\popguess'$ is universally quantified. \vspace{0.1cm}

    \item
        Skip 1:
        $\tup{\tup{\orcontrol, \control', \popguess'},
              \cha', \noop,
              \tup{\control'', \popguess, \rank}}$
        when
            $\popguess'$ is a vector $\tup{\controlset_0, \ldots, \controlset_\maxrank}$,
            $\control'' \in \controlset_{\rank}$.
            Recall that $\popguess$ is a component of $\orcontrol$.
            Note the quantification over $\rank$ and $\control''$.

    \item Skip 2:
        $\tup{\tup{\control'', \popguess, \rank},
              \cha', \noop,
              \tup{\control'', \popguess}}$.

    \item
        Push:
        $\tup{\tup{\orcontrol, \control', \popguess'},
              \cha', \noop,
              \tup{\control', \popguess'}}$.
    \end{itemize}

\item
    Character push: to simulate an order-$\cpdsord$ $\cpush{\chb}$, we use
    $\tup{\orcontrol,
          \cha', \cpush{\tup{\chb, \popguess}},
          \orcontrol'}$.
    Letter $\chb$ is annotated by the obligation $\popguess$ from $\orcontrol$.
    Indeed, a $\collapse$ eventually invoked on $\chb$ corresponds to an immediate $\pop{\cpdsord}$ and will return to the stack content below.
    Note, if
    $\popguess = \undefpopguess$
    a collapse would reach the empty stack and is thus prohibited.
    The control states that may be reached upon this return together with the least rank seen during the play on this stack level are precisely what is captured by $\popguess$.
\end{itemize}

When the original game returns from a recursive call, the order reduction has to evaluate whether the obligation
is met.
To this end, we use the following definitions.
It is here that we use the rank-awareness of $\cpds$.
\begin{itemize}
\item
    Pop: $\tup{\orcontrol, \cha', \noop, \controlsink}$
    when $\genop$ is $\pop{\cpdsord}$,
    $\popguess = \tup{\controlset_0, \ldots, \controlset_\maxrank}$,
    and $\controlsink = \ewin$ if
    $\control' \in \controlset_{\ap{\colrankfun}{\cpdsord}}$
    when
    $\cha = \tup{\cha'', \colrankfun}$.
    Otherwise we have $\controlsink = \elose$.
    Recall that $\popguess$ is a component of $\orcontrol$ and $\control'$ is from the simulated rule in $\cpdsrules$.

\item
    Collapse:
    $\tup{\orcontrol,
          \tup{\cha, \popguess'}, \noop,
          \controlsink}$
    when
        $\genop$ is $\collapse$,
        $\popguess'$ is the vector
        $\tup{\controlset_0, \ldots, \controlset_\maxrank}$, and
        $\controlsink = \ewin$ if
            $\control' \in \controlset_{\ap{\colrankfun}{\link}}$
            when
            $\cha = \tup{\cha'', \colrankfun}$.
        Otherwise $\controlsink = \elose$.
        Note the universal quantification over $\popguess'$
        (excluding $\undefpopguess$).
\end{itemize}

\begin{remark}[Rank-Awareness]
    The reduction of Hague\etal~\cite{HMOS08} maintains a rank $\rank$ in the control state.
    We do not need this $\rank$ since it is obtainable via rank-awareness.
    The original definition of rank-awareness did not strictly include this information, hence it could not be used.
    In this work, we observed that level-$\opord$ rank information could also be stored, as it appears in the proof of rank-awareness.

    Secondly, to be able to chain together the various reduction steps, we need the output of each to remain rank-aware.
    Thus, we apply \reflemma{lem:rank-aware} after each reduction to re-establish the property.
    Since the blow-up is polynomial, this is benign.
\end{remark}

\subsection{Correctness}

\begin{proposition}
\label{prop:orderred}
    \theeplayer wins $\game$ iff he wins
    $\ap{\orderred}{\game}$.
\end{proposition}
\begin{proof}
    When $\game$ is a parity game then
    $\ap{\orderred}{\game}$
    is also a parity game and the result is by Hague\etal~\cite{HMOS08}.
    When $\game$ is a safety game, the proof, \opcit, also shows equivalence.
    This is because a play in
    $\ap{\orderred}{\game}$
    is shown to be a factorization of a play in $\game$, where the minimal rank seen in a factor is faithfully recorded.
    Since, in a safety game, the minimal rank is $1$, then a violation of the safety condition will always be recognized.
\end{proof}

\section{Details of Counter Reduction}
\label{sec:appendix-counterreduction}

We complete the following definition in the sequel.

\begin{definition}[$\counterred{\bound}$]
    Given a bound $\bound\in \naturals$ and a parity game
    $\game = \tup{\cpds, \owner, \wincond}$
    over an order-$\cpdsord$ CPDS
    $\cpds = \tup{\controls,
                  \alphabet, \chorder,
                  \cpdsrules,
                  \controlinit, \chainit,
                  \rankf}$,
    we define a safety game
    $\ap{\counterred{\bound}}{\game} = \tup{\crcpds, \owner', 2^{\omega}}$
    where
    $\crcpds = \tup{\crcontrols,
                    \cralphabet, \crchorder,
                    \crcpdsrules,
                    \controlinit, \crchainit,
                    \crrankf}$.
\end{definition}

Let $\maxrank$ be the maximum rank in the game.

\subsection{Control States and Alphabet}

The control states are
\[
    \crcontrols\ =\ \controls \cup \brac{\controls \times [0, \maxrank]}
                 \cup
                 \set{\eover} \ .
\]
Control states $\tup{\control, \rank}$ indicate that a $\counterinc{\rank}$ operation needs to be performed.
The state $\eover$ is a sink state indicating \theeplayer has lost the game.
The owner of $\control \in \controls$ remains unchanged.
The owner of the states $\tup{\control, \rank}$ and $\eover$ does not matter, we give them to \theeplayer.
All control states except $\eover$ have rank $2$.
Control state $\eover$ has rank $1$.

The alphabet contains counter values for the odd ranks,
\[
    \cralphabet\ =\ \alphabet \times [1,\bound]^{\ceil{\maxrank / 2}}\ .
\]
The initial letter is $\crchainit= \tup{\chainit, \tup{0, \ldots, 0}}$ and
$\ap{\crchorder}{\cha, \countervec} = \ap{\chorder}{\cha}$.

\subsection{Rules}
The rules of $\crcpds$ proceed as in $\cpds$ but also maintain the counters.
We consider $\tup{\control_1, \cha, \genop, \control_2} \in \cpdsrules$ and list the rules in $\crcpdsrules$ by a case distinction along the operation.
To shorten the notation, we use $\rank_2$ for $\ap{\rankf}{\control_2}$.
Moreover, we let $\cha'=\tup{\cha, \countervec}$ and $\chb=\tup{\chb, \countervec}$.
Here, $\countervec$ is implicitly universally quantified but the same vector if both letters are used in the same rule.
Note that the pop and collapse simulations make use of the fact that we can assume $\cpds$ to be rank-aware.
That is, the stack character is $\cha = \tup{\cha'', \colrankfun}$. We have
\begin{itemize}
\item
    Push:
    $\tup{\control_1, \cha', \push{\opord}, \tup{\control_2, \rank_2}}$ when $\genop=\push{\opord}$.\vspace{0.1cm}

\item
    Character push:
    $\tup{\control_1, \cha', \cpush{\chb'}, \tup{\control_2, \rank_2}}$
    when
        $\genop= \cpush{\chb}$.\vspace{0.1cm}
\item
    Rewrite:
    $\tup{\control_1, \cha', \rew{\chb'}, \tup{\control_2, \rank_2}}$
    when $\genop=\rew{\chb}$.\vspace{0.1cm}

\item
    Pop:
    $\tup{\control_1, \cha', \pop{\opord}, \tup{\control_2, \rank'}}$
    when $\genop=\pop{\opord}$ and the rank is defined as
        $\rank' = \ap{\min}{\rank_2, \ap{\colrankfun}{\opord}}$.\vspace{0.1cm}

\item
    Collapse:
    $\tup{\control_1, \cha', \collapse, \tup{\control_2, \rank'}}$
    when $\genop=\collapse$ and the rank is defined as
            $\rank' = \ap{\min}{
                \rank_2,
                \ap{\colrankfun}{\link}
            }$.
\end{itemize}
The following rules perform the counter update.
For each state $\tup{\control, \rank}$ with current counter valuation $\countervec$, they determine $\ap{\counterinc{\rank}}{\countervec}$.
If no overflow occurs, the rules lead to $\control$ and the counter is updated with a rewriting operation.
Otherwise, they move to the sink state where \theeplayer loses.
\begin{itemize}
\item
    Update (suc):
    $\tup{\tup{\control, \rank},
          \tup{\cha, \countervec},
          \rew{\tup{\cha, \countervec'}},
          \control}$
    if
        $\ap{\counterinc{\rank}}{\countervec} = \countervec' \neq \overflow$.\vspace{0.1cm}

\item
    Update (fail):
    $\tup{\tup{\control, \rank},
          \tup{\cha, \countervec},
          \noop,
          \eover}$
    when
        $\ap{\counterinc{\rank}}{\countervec} = \overflow$.\vspace{0.1cm}
\item
    16 August 1977:
    $\tup{\eover, \cha, \noop, \eover}$.
\end{itemize}

\begin{remark}[Rank-Awareness]
    We finalise the reduction by applying Lemma~\ref{lem:rank-aware} to obtain a rank-aware CPDS.
\end{remark}

\section{Proof of Equivalence Between $\gameco$ and $\gameoc$}
\label{sec:equivalence-proof}

\subsection{Rank-Awareness}

The constructions of $\gameco$ and $\gameoc$ involve an intermediate rank-awareness transformation between the order-reduction and the counter-reduction (or vice-versa).
In the main text we have treated this operation as a black-box, meaning that it is hard to derive any relationship between the control states and stack characters of $\gameco$ and $\gameoc$.
As can be seen from Section~\ref{sec:order-1-rank-aware} the required construction is simply a product construction, where additional information is stored in the states/characters.
This information is merely maintained without having a restrictive effect on plays of the game.
That is, the resulting game is not fundamentally different from the original game: the plays and moves are the same, except an extra component of the control states and characters is updated during (or immediately after) each move.

Thus, to avoid notational clutter, we will omit this transformation from the following proof.
One may consider the rank-awareness information to be encoded as an implicit component of the control states and characters.

\subsection{From $\gameco$ to $\gameoc$}

Assume a winning strategy $\strat$ for \theeplayer in $\gameco$.
We demonstrate a winning strategy for \theeplayer in $\gameoc$.
Note, \theeplayer wins if they avoid reaching control state $\eover$ in $\gameco$ and any control state $\tup{\eover, \ldots}$ in $\gameoc$.

We assume without loss of generality that $\strat$ makes no guess
$\popguess = \tup{\controlset_0, \ldots, \controlset_\maxrank}$
that is not minimal.
That is, there is no winning strategy that could have guessed
$\popguess' = \tup{\controlset'_0, \ldots, \controlset'_\maxrank}$
with
$\controlset'_\rank \subseteq \controlset_\rank$
for all $\rank$ and the inclusion is strict in at least one $\rank$.

The initial configuration of both $\gameco$ and $\gameoc$ is
\[
    \config{
        \tup{\controlinit,
             \undefpopguess}
    }{
        \sbrac{\cdots\sbrac{
            \tup{\chainit, 0, \ldots, 0}
        }{1}\cdots}{\cpdsord}
    } \ .
\]

Hence, assume the play of $\gameco$ has reached some configuration
\[
    \config{
        \tup{\control_1, \popguess_1}
    }{
        \stackw_1
    }
\]
where
$\control_1$ is a control state of $\game$ and
$\popguess_1 = \undefpopguess$
or
$\popguess_1 = \tup{\controlset^1_0, \ldots, \controlset^1_\maxrank}$
where for all $\rank$ we have that
$\eover \notin \controlset^1_{\rank}$
and there is no
$\rank' \neq \rank$
and
$\control \in \controls$
with
$\tup{\control, \rank'} \in \controlset^1_{\rank}$.
Similarly, all characters
$\tup{\cha, \countervec, \popguess}$
appearing in $\stackw$ satisfy the same assumptions on $\popguess$.
Moreover, assume that the player of $\gameoc$ has reached a configuration
\[
    \config{
        \tup{\control_2, \popguess_2}
    }{
        \stackw_2
    }
\]
where
\begin{itemize}
\item
    $\control_1 = \control_2$,
\item
    $\popguess_1 = \popguess_2 = \undefpopguess$
    or
    $\popguess_2 = \tup{\controlset^2_0, \ldots, \controlset^2_\maxrank}$
    and for all $\rank$ we have
    $\control \in \controlset^2_\rank$
    iff
    $\tup{\control, \rank} \in \controlset^1_\rank$,
\item
    $\stackw_1$ and $\stackw_2$ differ only in the stack characters (not their structure -- i.e.~have the same stacks and links),
    in particular,
    \begin{itemize}
    \item
        if
        $\tup{\cha, \countervec}$
        appears in $\stackw_1$, then
        $\tup{\cha, \countervec}$
        also appears at the same position in $\stackw_2$.
    \item
        if
        $\tup{\cha, \countervec, \popguess}$
        appears in $\stackw_1$, then
        $\tup{\cha, \countervec, \popguess'}$
        appears at the same position in $\stackw_2$ and for all $\rank$ we have
        $\control \in \controlset'_\rank$
        iff
        $\tup{\control, \rank} \in \controlset'_\rank$
        where
        $\popguess = \tup{\controlset_0, \ldots, \controlset_\maxrank}$
        and
        $\popguess' = \tup{\controlset'_0, \ldots, \controlset'_\maxrank}$
        (or both are $\undefpopguess$).
    \end{itemize}
\end{itemize}

We assume it is \theaplayer who is to move in $\gameoc$.
We show that whichever move they make, there is a corresponding move in $\gameco$ and the invariant is maintained once play returns to a control state of the above form.
In the case that it is \theeplayer who is to move, the argument can easily be adapted by taking the move prescribed by $\strat$ in $\gameco$ and playing the corresponding move in $\gameoc$.
For this, we need that there are no additional moves in $\gameco$ which can easily be verified by inspecting the proof.
Indeed, this case is almost identical to the proof for moves by \theaplayer when arguing the opposite direction (from $\gameoc$ to $\gameco$), which is given below.

Thus, we case split on the move played by \theaplayer in $\gameoc$.
\begin{itemize}
\item
    Move
    $\tup{
        \tup{\control_2, \popguess_2},
        \cha,
        \genop,
        \tup{\control'_2, \popguess_2}
     }$
    where $\control'_2$ is a control state of $\game$.

    In this case there is also a move
    \[
        \tup{
           \tup{\control_1, \popguess_1},
           \cha',
           \genop',
           \tup{\control'_1, \popguess_1}
        }
    \]
    in $\gameco$ which maintains
    $\control'_1 = \control'_2$.
    There are several further cases depending on $\genop$.

    \begin{itemize}
    \item
        If $\genop = \rew{\chb}$ for some $\chb$ then if
        $\chb = \tup{\chb', \countervec, \popguess}$
        we have
        $\genop' = \rew{\tup{\chb', \countervec, \popguess'}}$
        where the required properties on $\popguess$ and $\popguess'$ follow from the fact that they also appeared in $\cha$ and $\cha'$ and our assumption
        (I.e.~$\cha$ and $\cha'$ are triples whose third component is $\popguess$ and $\popguess'$ respectively).
        Moreover, the counter values are updated identically.
        If
        $\chb = \tup{\chb', \countervec}$
        the argument is simpler.

    \item
        If
        $\genop = \cpush{\chb}$
        with
        $\chb = \tup{\chb', \countervec}$
        then
        $\genop' = \cpush{\chb}$
        and the invariant is maintained.

    \item
        If
        $\genop = \cpush{\chb}$
        with
        $\chb = \tup{\chb', \countervec, \popguess}$
        then
        $\popguess = \popguess_2$
        and we have
        $\genop' = \cpush{\tup{\chb', \countervec, \popguess'}}$
        with
        $\popguess' = \popguess_1$
        and the invariant is maintained.

    \item
        Otherwise $\genop$ adds no new characters to the stack (only copies) and $\genop' = \genop$ and the invariant is easily maintained.
    \end{itemize}

\item
    Move
    $\tup{
        \tup{\control_2, \popguess_2},
        \cha,
        \genop,
        \tup{\tup{\control'_2, \popguess_2}, \rank}
     }$
    where $\genop = \pop{\opord}$ or $\genop = \collapse$.

    In this case the corresponding move in $\gameco$ is
    \[
        \tup{
            \tup{\control_1, \popguess_1},
            \cha',
            \genop,
            \tup{\tup{\control'_1, \rank}, \popguess_1}
        }
    \]
    where $\cha = \cha'$ if
    $\cha = \tup{\chb, \countervec}$
    and otherwise
    $\cha = \tup{\chb, \countervec, \popguess}$
    and
    $\cha' = \tup{\chb, \countervec, \popguess'}$.
    Moreover
    $\control'_1 = \control'_2$.

    From the resulting configurations, there is only one move in each game.
    Note, the counters do not overflow since $\strat$ is winning.
    Assume
    $\cha = \tup{\chb, \countervec, \popguess}$
    (the other case is simpler).
    In $\gameoc$ the move \\
    $\tup{
        \tup{\tup{\control_2, \popguess_2}, \rank},
        \tup{\chb, \countervec, \popguess},
        \rew{\tup{\chb, \countervec', \popguess}},
        \tup{\control'_2, \popguess_2}
    }$ \\
    is applied where
    $\countervec' = \ap{\counterinc{\rank}}{\countervec}$.
    In $\gameco$ we have the move \\
    $\tup{
        \tup{\tup{\control_1, \rank}, \popguess_1},
        \tup{\chb, \countervec, \popguess'},
        \rew{\tup{\chb, \countervec', \popguess'}},
        \tup{\control'_1, \popguess_1}
     }$ \\
    and the invariant is maintained.

\item
    Move
    $\tup{
        \tup{\control_2, \popguess_2},
        \cha,
        \noop,
        \tup{\control_2, \popguess_2, \control}
     }$.

    In this case, $\control_2$ and $\control$ is a control state of $\game$ and there is a move in $\gameco$ of the form
    \[
        \tup{
            \tup{\control_1, \popguess_1},
            \cha,
            \noop,
            \tup{\control_1, \popguess_1, \control}
         } \ .
    \]

    In both cases, play reaches a state where \theeplayer must guess a sequence of sets of control states.
    Suppose $\strat$ in $\gameco$ recommends the move
    \[
        \tup{
            \tup{\control_1, \popguess_1, \control},
            \cha,
            \noop,
            \tup{\control_1, \popguess_1, \control, \popguess'_1}
        } \ .
    \]
    we note the following about
    $\popguess'_1 = \tup{\controlset_0, \ldots, \controlset_\maxrank}$
    (which is not $\undefpopguess$).
    \begin{itemize}
    \item
        Suppose
        $\tup{\control', \rank'} \in \controlset_\rank$
        with $\rank' \neq \rank$.
        Then consider the play from
        $\tup{\control, \ap{\rankf}{\control}, \popguess'_1}$.
        Suppose \theaplayer can force play to simulate a $\pop{\cpdsord}$ or $\collapse$ (for some order-$\cpdsord$ link) which relies on
        $\tup{\control', \rank'} \in \controlset_\rank$
        for \theeplayer to win.
        In this case, since by the definition of $\counterred{\bound}$ a state
        $\tup{\control', \rank'}$
        will only be reached if the smallest rank seen since the corresponding push is $\rank'$.
        However, for the victory of \theeplayer to depend on
        $\tup{\control', \rank'} \in \controlset_\rank$
        then the smallest rank seen is $\rank \neq \rank'$, which is a contradiction.
        Hence $\strat$ could have removed
        $\tup{\control', \rank'}$ from  $\controlset_\rank$.
        Since we have assumed $\strat$ is minimal, it must be the case that
        $\tup{\control', \rank'} \notin \controlset_\rank$.

    \item
        Suppose
        $\control' \in \controlset_\rank$.
        Then, we observe that the definition of
        $\counterred{\bound}$
        does not create any pop moves to $\control'$
        (only to $\tup{\control', \rank'}$ for some $\rank'$).
        Thus, $\strat$ could remove $\control'$ from $\controlset_\rank$ and still be winning.
        Since we assumed $\strat$ is minimal, we know
        $\control' \notin \controlset_\rank$.
    \end{itemize}
    Thus, in $\controlset_\rank$ we only have states of the form
    $\tup{\control, \rank}$.
    We play the move
    \[
        \tup{
            \tup{\control_2, \popguess_2, \control},
            \cha,
            \noop,
            \tup{\control_2, \popguess_2, \control, \popguess'_2}
        }
    \]
    where
    $\popguess'_2 = \tup{\controlset'_0, \ldots, \controlset'_\maxrank}$
    and for each $\rank$ we set
    $\controlset'_\rank$
    to the set
    $\setcomp{\control'}{\tup{\control', \rank} \in \controlset_\rank}$.

    Now there are two cases where \theaplayer can move to in $\gameoc$.
    \begin{itemize}
    \item
        If play moves to
        $\tup{\control', \popguess_2, \rank}$
        and then
        $\tup{\control', \popguess_2}$
        for some
        $\control' \in \controlset'_\rank$
        then play could also move to
        $\tup{\control', \popguess_1, \rank}$
        and then
        $\tup{\control', \popguess_1}$
        in $\gameco$ and the invariant is satisfied.

    \item
        If play moves to
        $\tup{\control, \popguess'_2}$
        then play could also move to
        $\tup{\control, \popguess'_1}$
        in $\gameco$ and the invariant is satisfied.
    \end{itemize}
\end{itemize}

To complete the proof, observe that \theeplayer can win the safety game by maintaining the proven invariant, since these plays never reach a losing state.

\subsection{From $\gameoc$ to $\gameco$}

Assume a winning strategy $\strat$ for \theeplayer in $\gameoc$.
We construct a winning strategy for \theeplayer in $\gameco$.
As before, the initial configuration of both games is
\[
    \config{
        \tup{\controlinit,
             \undefpopguess}
    }{
        \sbrac{\cdots\sbrac{
            \tup{\chainit, 0, \ldots, 0}
        }{1}\cdots}{\cpdsord}
    } \ .
\]

Assume that the player of $\gameoc$ has reached a configuration
\[
    \config{
        \tup{\control_2, \popguess_2}
    }{
        \stackw_2
    }
\]
where
$\popguess_2 = \tup{\controlset^2_0, \ldots, \controlset^2_\maxrank}$
(or $\undefpopguess$)
and that the play of $\gameco$ has reached some configuration
\[
    \config{
        \tup{\control_1, \popguess_1}
    }{
        \stackw_1
    }
\]
where
$\popguess_1 = \tup{\controlset^1_0, \ldots, \controlset^1_\maxrank}$
(or $\undefpopguess$)
and
\begin{itemize}
\item
    $\control_1 = \control_2$,
\item
    either
    $\popguess_1 = \popguess_2 = \undefpopguess$
    or neither are $\undefpopguess$ and for all $\rank$ we have
    $\tup{\control, \rank} \in \controlset^1_\rank$,
    iff
    $\control \in \controlset^2_\rank$
\item
    $\stackw_1$ and $\stackw_2$ differ only in the stack characters (not their structure -- i.e.~have the same stacks and links),
    in particular,
    \begin{itemize}
    \item
        if
        $\tup{\cha, \countervec}$
        appears in $\stackw_1$, then
        $\tup{\cha, \countervec}$
        also appears at the same position in $\stackw_2$.
    \item
        if
        $\tup{\cha, \countervec, \popguess}$
        appears in $\stackw_1$, then
        $\tup{\cha, \countervec, \popguess'}$
        appears at the same position in $\stackw_2$ and either
        $\popguess = \popguess' = \undefpopguess$
        or for all $\rank$ we have
        $\control \in \controlset'_\rank$
        iff
        $\tup{\control, \rank} \in \controlset'_\rank$
        where
        $\popguess = \tup{\controlset_0, \ldots, \controlset_\maxrank}$
        and
        $\popguess' = \tup{\controlset'_0, \ldots, \controlset'_\maxrank}$.
    \end{itemize}
\end{itemize}

We will consider moves by \theaplayer in $\gameco$ and show that there is a corresponding move in $\gameoc$.
That is, however \theaplayer continues the game in $\gameco$, the winning strategy for $\gameoc$ remains an adequate guide for winning $\gameco$.
When it is the turn of \theeplayer to move, we have to show that the move chosen by \theeplayer in $\gameco$ can be mimicked in $\gameco$.
This is straightforward to check following similar arguments as for \theaplayer case in the other direction.

There are several cases depending on the move by \theaplayer in $\gameco$.
\begin{itemize}
\item
    Move
    $\tup{
        \tup{\control_l, \popguess_1},
        \cha,
        \genop,
        \tup{\control'_1, \popguess_1}
     }$
    where $\control'_1$ is a control state of $\game$.

    In this case there is also a move
    \[
        \tup{
           \tup{\control_2, \popguess_2},
           \cha,
           \genop',
           \tup{\control'_2, \popguess_2}
        }
    \]
    in $\gameoc$ which maintains
    $\control'_1 = \control'_2$.
    There are several further cases depending on $\genop$.

    \begin{itemize}
    \item
        If $\genop = \rew{\chb}$ for some $\chb$ then if
        $\chb = \tup{\chb', \countervec, \popguess}$
        we have
        $\genop' = \rew{\tup{\chb', \countervec, \popguess'}}$
        where the required properties on $\popguess$ and $\popguess'$ follow from the fact that they also appeared in $\cha$ and our assumption
        (I.e.~$\cha$ and $\cha'$ are triples whose third component is $\popguess$ and $\popguess'$ respectively).
        Moreover, the counter values are updated identically.
        If
        $\chb = \tup{\chb', \countervec}$
        the argument is simpler.

    \item
        If
        $\genop = \cpush{\chb}$
        with
        $\chb = \tup{\chb', \countervec}$
        then
        $\genop' = \cpush{\chb}$
        and the invariant is maintained.

    \item
        If
        $\genop = \cpush{\chb}$
        with
        $\chb = \tup{\chb', \countervec, \popguess}$
        then
        $\popguess = \popguess_1$
        and we have
        $\genop' = \cpush{\tup{\chb', \countervec, \popguess'}}$
        with
        $\popguess' = \popguess_2$
        and the invariant is maintained.

    \item
        Otherwise $\genop$ adds no new characters to the stack (only copies) and $\genop' = \genop$ and the invariant is easily maintained.
    \end{itemize}

\item
    Move
    $\tup{
        \tup{\control_1, \popguess_1},
        \cha',
        \genop,
        \tup{\tup{\control'_1, \rank}, \popguess_1}
    }$
    where $\genop = \pop{\opord}$ or $\genop = \collapse$.

    In this case the corresponding move in $\gameoc$ is
    \[
        \tup{
            \tup{\control_2, \popguess_2},
            \cha,
            \genop,
            \tup{\tup{\control'_2, \popguess_2}, \rank}
        }
    \]
    where $\cha = \cha'$ if
    $\cha = \tup{\chb, \countervec}$
    and otherwise
    $\cha = \tup{\chb, \countervec, \popguess}$
    and
    $\cha' = \tup{\chb, \countervec, \popguess'}$.
    Moreover
    $\control'_1 = \control'_2$.

    From the resulting configurations, there is only one move in each game.
    Note, the counters do not overflow since $\strat$ is winning.
    Assume
    $\cha = \tup{\chb, \countervec, \popguess}$
    (the other case is simpler).
    In $\gameco$ we have the move
    \[
        \tup{
            \tup{\tup{\control_1, \rank}, \popguess_1},
            \tup{\chb, \countervec, \popguess'},
            \rew{\tup{\chb, \countervec', \popguess'}},
            \tup{\control'_1, \popguess_1}
         }
    \]
    is applied where
    $\countervec' = \ap{\counterinc{\rank}}{\countervec}$.
    In $\gameoc$ the move
    \[
        \tup{
            \tup{\tup{\control_2, \popguess_2}, \rank},
            \tup{\chb, \countervec, \popguess},
            \rew{\tup{\chb, \countervec', \popguess}},
            \tup{\control'_2, \popguess_2}
        }
    \]
    and the invariant is maintained.

\item
    Move
    $\tup{
        \tup{\control_1, \popguess_1},
        \cha,
        \noop,
        \tup{\control_1, \popguess_1, \control}
     }$.

    In this case, $\control_1$ and $\control$ is a control state of $\game$ and there is a move in $\gameoc$ of the form
    \[
        \tup{
            \tup{\control_2, \popguess_2},
            \cha,
            \noop,
            \tup{\control_2, \popguess_2, \control}
         } \ .
     \]

    In both cases, play reaches a state where \theeplayer must guess a sequence of sets of control states.
    Suppose $\strat$ in $\gameoc$ recommends the move
    \[
        \tup{
            \tup{\control_2, \popguess_2, \control},
            \cha,
            \noop,
            \tup{\control_2, \popguess_2, \control, \popguess'_2}
        } \ .
    \]
    In $\controlset_\rank$ we have states of the form $\control$ where $\control$ is a control state of $\game$.
    In $\gameco$ we play the move
    \[
        \tup{
            \tup{\control_1, \popguess_1, \control},
            \cha,
            \noop,
            \tup{\control_1, \popguess_1, \control, \popguess'_1}
         }
    \]
    where
    $\popguess'_1 = \tup{\controlset'_0, \ldots, \controlset'_\maxrank}$
    and for each $\rank$ we set
    \[
        \controlset'_\rank =
            \setcomp{\tup{\control', \rank}}
                    {\control' \in \controlset_\rank} \ .
    \]

    Now there are two cases where \theaplayer can move to in $\gameco$.
    \begin{itemize}
    \item
        If play moves to
        $\tup{\control', \popguess_1, \rank}$
        and then
        $\tup{\control', \popguess_1}$
        for some
        $\control' \in \controlset'_\rank$
        then play could also move to
        $\tup{\control', \popguess_2, \rank}$
        and then
        $\tup{\control', \popguess_2}$
        in $\gameoc$ and the invariant is satisfied.

    \item
        If play moves to
        $\tup{\control, \popguess'_1}$
        then play could also move to
        $\tup{\control, \popguess'_2}$
        in $\gameoc$ and the invariant is satisfied.
    \end{itemize}
\end{itemize}

To complete the proof, observe that \theeplayer can win the safety game by maintaining the proven invariant, since these plays never reach a losing state.

\section{Correctness of the Polynomial Reduction}
\label{sec:poly-equiv-proof}

We provide the proof of \reflemma{lem:poly-equiv}.
First note, since \theeplayer wins
$\ap{\counterred{\bound}}{\game}$
iff they win
$\ap{\counterred{\bound'}}{\game}$
for any $\bound' \geq \bound$ we can select $\bound$ such that Cachat and Walukiewicz's counter encoding precisely captures numbers between $0$ and $\bound$.

We also define $\polyred{\bound}$ on stack characters with links.
That is
\[
    \ap{\polyred{\bound}}{
        \annot{
            \tup{\cha, \counter_1, \ldots, \counter_\maxrank}
        }{
            \idxi
        }
    } =
    \cha
    \stackw_\maxrank
    \ldots
    \stackw_1
    \annot{\linkchar{\opord}}{\idxi}
\]
where the stack on the right has its topmost character on the left.
All characters except $\linkchar{\opord}$ have an order-1 link which is omitted.
Moreover, for each odd
$1 \leq \rank \leq \maxrank$
we have that $\stackw_\rank$ is the Cachat and Walukiewicz encoding of $\counter_\rank$.
Let
\[
    \ap{\polyred{\bound}}{\control, \stackw}
        = \tup{\control, \stackw'}
\]
where $\stackw'$ is obtained from $\stackw$ by the pointwise application of
$\polyred{\bound}$
to each character in the stack.

In both directions below we will maintain the invariant that if
$\ap{\counterred{\bound}}{\game}$
is in configuration
$\tup{\control, \stackw}$
then
$\ap{\polyred{\bound}}{\game}$
is in configuration
$\ap{\polyred{\bound}}{\control, \stackw}$.

We first consider the behaviour for important subgames of the game
$\ap{\polyred{\bound}}{\control, \stackw}$.
In the following, we do not show unimportant links (to avoid notational clutter).
Moreover, for a sequence of characters with links
$\annot{\cha_1}{\idxi_1}
 \ldots
 \annot{\cha_\numof}{\idxi_\numof}$
we abuse notation and write
\[
    \ccompose{
        \annot{\cha_1}{\idxi_1}
        \ldots
        \annot{\cha_\numof}{\idxi_\numof}
    }{1}{\stackw}
\]
in place of
$\ccompose{\annot{\cha_1}{\idxi_1}}{1}{
    \ccompose{\cdots}{1}{
        \ccompose{\annot{\cha_\numof}{\idxi_\numof}}{1}{
            \stackw
        }
    }
 }$.

\subsection{Behaviour of Subgames}

We argue that certain subgames of
$\ap{\polyred{\bound}}{\game}$
implement specific desired behaviour.

\begin{itemize}
\item
    From a configuration
    $\tup{\controlzero{\rank}{\cha}{\control}, \stackw}$
    either \theeplayer wins the game or play reaches
    $\tup{\control, \ccompose{\stackw'}{1}{\stackw}}$
    where the stack $\stackw'$ is
    $\cha \stackw_\maxrank \ldots \stackw_\rank$
    and for each odd
    $\maxrank \geq \rank' \geq \rank$
    we have that
    $\stackw_{\rank'}$
    is the Cachat and Walukiewicz encoding of $0$.

    We assume that the topmost characters of $\stackw$ are from
    $\encodingalphabet{\rank} \cup \set{\bzero{\rank}}$.
    Let
    $\stackw = \ccompose{\stacku}{1}{\stackw'}$
    where $\stackw'$ contains no characters from
    $\counteralphabet{\rank}$.

    This configuration belongs to \theaplayer and they may push as many characters from
    $\encodingalphabet{\rank} \cup \set{\bzero{\rank}}$
    as they wish onto the stack.
    If they push forever, they will never reach an unsafe configuration and \theeplayer will win.
    The only other move is to
    \[
        \tup{\controlzerotest{\rank}{\cha}{\control},
             \ccompose{\stacku'}{1}{\stackw'}}
    \]
    where
    $\stacku' \in \brac{\encodingalphabet{\rank} \cup \set{\bzero{\rank}}}^\ast$.

    We know from the above that $\stacku'$ does not contain $\bone{\rank}$.
    At this configuration it is the move of \theeplayer.
    If $\stacku'$ does not encode a valid counter, \theeplayer moves to
    $\countercheck{\rank}$
    and \theaplayer loses the game
    (this includes the case where \theaplayer did not push anything).
    Otherwise, playing to
    $\countercheck{\rank}$
    would cause \theeplayer to lose.

    We now argue by induction.
    In the base case, let $\rank = \maxrank$.
    We know $\stacku'$ is a valid counter using only $\bzero{\maxrank}$ and \theeplayer moves to
    $\tup{\control, \ccompose{\cha\stacku'}{1}{\stackw'}}$
    as required.
    If $\rank < \maxrank$ then \theeplayer moves to
    $\tup{\controlzero{\rank+2}{\cha}{\control},
          \ccompose{\stacku'}{1}{\stackw'}}$
    and by induction, either \theeplayer wins or play reaches
    $\tup{\control, \ccompose{\stackw''}{1}{\stackw}}$
    where
    $\stackw'' = \cha \stackw_\maxrank \ldots \stackw_\rank$
    as required.

\item
    From a configuration
    $\tup{\controlcopy{\rank}{\cha}{\control}, \stackw'}$
    with
    $\stackw' =$
    \[
         \ccompose{\stacku
                   \stackw_{\rank-2}\ldots\stackw_1
                   \annot{\linkchar{\opord}}{\idxi}
                   \chb
                   \stackw_\maxrank\ldots\stackw_1
                   \annot{\linkchar{\opord'}}{\idxi'}}
                  {1}
                  {\stackw}
    \]
    where
    $\stackw_{\rank'} \in \counteralphabet{\rank'}^\ast$
    for each odd $\rank'$ and
    $\stacku \in \counteralphabet{\rank}^\ast$,
    either \theeplayer wins the game or play reaches the configuration
    $\tup{\controlincnoa{\ap{\rankf}{\control}}{\control}, \stackw''}$
    with
    $\stackw'' =$
    \[
         \ccompose{\stackw_\maxrank\ldots\stackw_1
                   \annot{\linkchar{\opord}}{\idxi}
                   \chb
                   \stackw_\maxrank\ldots\stackw_1
                   \annot{\linkchar{\opord'}}{\idxi'}}
                  {1}
                  {\stackw} \ .
    \]
    The argument is almost identical to the above case.
    The key difference is that \theaplayer can push any character from
    $\counteralphabet{\rank}$
    including $\bone{\rank}$.
    They have the option to push forever (and lose) or move to
    $\controlcopytest{\rank}{\cha}{\control}$.
    If they have not copied $\stackw_\rank$ correctly, \theeplayer can move to
    $\equalscheck{\rank}$
    and win the game.
    If they have copied $\stackw_\rank$ correctly then \theeplayer must move to either
    \begin{itemize}
    \item
        $\tup{\controlincnoa{\ap{\rankf}{\control}}{\control},
              \stackw''}$
        (where $\stackw''$ is the target stack defined above)
        if $\rank = \maxrank$, or

    \item
        $\tup{\controlcopy{\rank+2}{\cha}{\control},
              \stackw'''}$
        where
        $\stackw''' =$
        \[
            \ccompose{\stackw_{\rank}\ldots\stackw_1
                      \annot{\linkchar{\opord}}{\tup{\opord,\idxi}}
                      \chb
                      \stackw_\maxrank\ldots\stackw_1
                      \annot{\linkchar{\opord'}}{\tup{\opord',\idxi'}}}
                     {1}
                     {\stackw}
        \]
    \end{itemize}
    or lose the game.
    In the former case, we have reached the required configuration.
    In the latter, we will reach the required configuration by induction.

\item
    From a configuration
    \[
        \tup{\controlincnoa{\rank}{\control},
             \ccompose{\cha
                       \stackw_\maxrank \ldots \stackw_1
                       \annot{\linkchar{\opord}}{\idxi}}
                      {1}
                      {\stackw}}
    \]
    where there are counters
    $\counter_1,\ldots,\counter_\maxrank$
    such that
    \[
        \ap{\polyred{\bound}}{
            \annot{\tup{\cha, \counter_1, \ldots, \counter_\maxrank}}
                  {\idxi}
        }
        =
        \cha
        \stackw_\maxrank \ldots \stackw_1
        \annot{\linkchar{\opord}}{\idxi}
    \]
    then
    \begin{itemize}
    \item
        if
        $\tup{\counter'_1, \ldots, \counter'_\maxrank}
         =
         \ap{\counterinc{\rank}}{\counter_1, \ldots, \counter_\maxrank}$
        then \theeplayer either wins the game or play reaches
        \[
            \tup{\control,
                 \ccompose{\cha
                           \stackw'_\maxrank \ldots \stackw'_1
                           \annot{\linkchar{\opord}}{\idxi}}
                          {1}
                          {\stackw}}
        \]
        where
        \[
            \ap{\polyred{\bound}}{
                \annot{\tup{\cha, \counter'_1, \ldots, \counter'_\maxrank}}
                      {\idxi}
            }
            =
            \cha
            \stackw'_\maxrank \ldots \stackw'_1
            \annot{\linkchar{\opord}}{\idxi}
        \]
        and

    \item
        if
        $\ap{\counterinc{\rank}}{\counter_1, \ldots, \counter_\maxrank}
            = \overflow$
        then \theeplayer loses the game.
    \end{itemize}

    The first move is necessarily to
    \[
        \tup{\controlinc{\rank}{\cha}{\control},
             \ccompose{\stackw_\maxrank \ldots \stackw_1
                       \annot{\linkchar{\opord}}{\idxi}}
                      {1}
                      {\stackw}}
    \]
    From here it is only possible to pop from the stack until play reaches
    \[
        \tup{\controlinc{\rank}{\cha}{\control},
             \ccompose{\stackw_{\rank'} \ldots \stackw_1
                       \annot{\linkchar{\opord}}{\idxi}}
                      {1}
                      {\stackw}}
    \]
    where $\rank = \rank'$ if $\rank$ is odd, else $\rank = \rank'-1$.

    If $\rank$ is even the only possible move is to
    \[
        \tup{\controlzero{\rank+1}{\cha}{\control},
             \ccompose{\stackw_{\rank'} \ldots \stackw_1
                       \annot{\linkchar{\opord}}{\idxi}}
                      {1}
                      {\stackw}}
    \]
    from which, as argued above, either \theeplayer wins the game or play reaches
    \[
        \tup{\control,
             \ccompose{\cha\stackw'_\maxrank \ldots \stackw'_{\rank+1}
                       \stackw_{\rank'} \ldots \stackw_1
                       \annot{\linkchar{\opord}}{\idxi}}
                      {1}
                      {\stackw}}
    \]
    where each $\stackw'_{\rank''}$ encodes the number $0$.
    In this case, the counters have been correctly implemented as required.

    If $\rank$ is odd then, since the top of the stack has characters from
    $\counteralphabet{\rank}$
    it is only possible to perform $\pop{1}$ operations until the first $\bzero{\rank}$ is found.

    If
    $\ap{\counterinc{\rank}}{\counter_1, \ldots, \counter_\maxrank}
        = \overflow$
    then there is no such $\bzero{\rank}$ and eventually all characters from
    $\counteralphabet{\rank}$
    are removed and play reaches a configuration
    $\tup{\controlinc{\rank}{\cha}{\control}, \stackw''}$
    where the topmost character of $\stackw''$ is either from
    $\counteralphabet{\rank-2}$
    or is $\linkchar{\opord}$.
    In this case play moves to the control state $\elose$ and \theeplayer loses the game as required.

    Otherwise, play reaches
    \[
        \tup{\controlinc{\rank}{\cha}{\control},
             \ccompose{\bzero{\rank}\stacku
                       \stackw_{\rank-2} \ldots \stackw_1
                       \annot{\linkchar{\opord}}{\idxi}}
                      {1}
                      {\stackw}}
    \]
    where $\stacku$ is the suffix of $\stackw_\rank$ after the first $\bzero{\rank}$.
    From here the only possible move is to
    \[
        \tup{\controlzero{\rank}{\cha}{\control},
             \ccompose{\bone{\rank}\stacku
                       \stackw_{\rank-2} \ldots \stackw_1
                       \annot{\linkchar{\opord}}{\idxi}}
                      {1}
                      {\stackw}}
    \]
    and \theaplayer can push any sequence of characters from
    $\encodingalphabet{\rank} \cup \set{\bzero{\rank}}$.
    If they push forever, then \theeplayer wins the game.
    Otherwise, they must reach
    \[
        \tup{\controlzerotest{\rank}{\cha}{\control},
             \ccompose{\stacku'\bone{\rank}\stacku
                       \stackw_{\rank-2} \ldots \stackw_1
                       \annot{\linkchar{\opord}}{\idxi}}
                      {1}
                      {\stackw}}
    \]
    where
    $\stacku' \in \brac{\encodingalphabet{\rank} \cup \set{\bzero{\rank}}}^\ast$.
    If
    $\stacku'\bone{\rank}\stacku$
    is not a valid counter encoding, then \theeplayer moves to
    $\countercheck{\rank}$
    and wins the game.
    Otherwise, the counter encoding is valid and by design, the value denoted by
    $\stacku'\bone{\rank}\stacku$
    must be the increment of the value denoted by
    $\stackw_\rank$.

    If $\rank = \maxrank$ the only possible move is to
    \[
        \tup{\control, \cha
                       \stackw'_\maxrank
                       \stackw_{\maxrank-2}
                       \ldots
                       \stackw_1
                       \annot{\linkchar{\opord}}{\idxi}}
    \]
    where
    $\stackw'_\maxrank = \stacku'\bone{\rank}\stacku$
    and the configuration can be seen to satisfy the requirements above.

    Finally, if $\rank < \maxrank$ then the only possible move is to
    \[
        \tup{\controlzero{\rank+2}{\cha}{\control},
             \ccompose{\stackw'_{\rank} \ldots \stackw_1
                       \annot{\linkchar{\opord}}{\idxi}}
                      {1}
                      {\stackw}}
    \]
    where
    $\stackw'_\maxrank = \stacku'\bone{\rank}\stacku$
    and from which, as argued above, either \theeplayer wins the game or play reaches the configuration
    $\tup{\control,
          \ccompose{\cha\stackw'_\maxrank \ldots \stackw'_{\rank}
                    \stackw_{\rank-2} \ldots \stackw_1
                    \annot{\linkchar{\opord}}{\idxi}}
                   {1}
                   {\stackw}}$
    where each $\stackw'_{\rank''}$ for $\rank'' > \rank$ encodes the number $0$.
    In this case as well, the counters have been correctly implemented as required.
\end{itemize}

\subsection{%
    From
    $\ap{\counterred{\bound}}{\game}$
    to
    $\ap{\polyred{\bound}}{\game}$
}

Suppose $\strat$ is a winning strategy for \theeplayer in the game
$\gamebound = \ap{\counterred{\bound}}{\game}$.
We define a winning strategy $\strat'$ for \theeplayer in
$\gamepoly = \ap{\polyred{\bound}}{\game}$.
The strategy is to essentially copy the winning play from
$\ap{\counterred{\bound}}{\game}$.

The initial configuration in $\gamebound$ is
\[
    \tup{\controlinit,
         \sbrac{\cdots\sbrac{
             \tup{\chainit, 0, \ldots, 0}
         }{1}\cdots}{\cpdsord}} \ .
\]
The initial configuration in $\gamepoly$ is
$\tup{\controlinit',
      \sbrac{\cdots\sbrac{
          \chainit'
      }{1}\cdots}{\cpdsord}}$
from which play is forced to
\[
    \tup{\controlzero{1}{\chainit}{\control},
         \sbrac{\cdots\sbrac{
             \linkchar{1}
         }{1}\cdots}{\cpdsord}}
\]
and, by the previous section, either \theeplayer wins, or play reaches
$\tup{\control,
      \chainit
      \stackw_\maxrank \ldots \stackw_1
      \linkchar{1}}$
where each $\stackw_\rank$ encodes $0$.
Thus, we have set up the invariant that when $\gamebound$ is in
$\tup{\control, \stackw}$
then $\gamepoly$ is in
$\ap{\polyred{\bound}}{\control, \stackw}$.

We show that for every pair of moves played in $\gamebound$ there is a play in $\gamepoly$ that is either winning for \theeplayer or re-establishes the invariant.
There are a number of cases.
Note, counters never overflow because $\strat$ is winning for \theeplayer.
\begin{itemize}
\item
    Whenever we have
    $\tup{\control,
          \tup{\cha, \counter_\maxrank, \ldots, \counter_1},
          \push{\opord},
          \tup{\control', \rank}}$
    followed by
    \[
        \tup{\tup{\control', \rank},
             \tup{\cha, \counter_\maxrank, \ldots, \counter_1},
             \rew{\tup{\cha, \counter'_\maxrank, \ldots, \counter'_1}},
             \control'}
    \]
    (with
     $\rank = \ap{\rankf}{\control'}$)
    in $\gamebound$, then in $\gamepoly$ we have
    \[
        \tup{\control, \cha, \push{\opord}, \controlinc{\rank}{\cha}{\control'}}
    \]
    after which, as argued in the previous section, either \theeplayer wins or play reaches $\control'$ with the counters encoded on the stack incremented correctly.

\item
    Whenever we have
    \[
        \tup{\control,
             \tup{\cha, \counter_1, \ldots, \counter_\maxrank},
             \cpush{\tup{\chb, \counter_1, \ldots, \counter_\maxrank}},
             \tup{\control', \rank}}
    \]
    followed by
    \[
        \tup{\tup{\control', \rank},
             \tup{\chb, \counter_1, \ldots, \counter_\maxrank},
             \rew{\tup{\chb, \counter'_1, \ldots, \counter'_\maxrank}},
             \control'}
    \]
    in $\gamebound$ then in $\gamepoly$ we have
    \[
        \tup{\control,
             \cha,
             \cpush{\linkchar{\opord}},
             \controlcopy{\rank}{\chb}{\control'}} \ .
    \]
    As argued in the previous section, \theeplayer wins or play must reach
    $\controlincnoa{\rank}{\control'}$
    after pushing
    $\chb \stackw_\maxrank \ldots \stackw_1$
    encoding the counters
    $\counter_1, \ldots, \counter_\maxrank$
    and then from
    $\controlincnoa{\rank}{\control'}$
    play must reach
    $\control$
    with the counter encodings incremented to match
    $\counter'_1, \ldots, \counter'_\maxrank$
    as required.

\item
    Whenever we have
    \[
        \tup{\control,
             \tup{\cha, \counter_\maxrank, \ldots, \counter_1},
             \rew{\tup{\chb, \counter_1, \ldots, \counter_\maxrank}},
             \tup{\control', \rank}}
    \]
    followed by
    \[
        \tup{\tup{\control', \rank},
             \tup{\chb, \counter_\maxrank, \ldots, \counter_1},
             \rew{\tup{\chb, \counter'_\maxrank, \ldots, \counter'_1}},
             \control'}
    \]
    in $\gamebound$, then in $\gamepoly$ we have
    $\tup{\control, \cha, \rew{\chb}, \controlinc{\rank}{\cha}{\control'}}$
    after which, as argued in the previous section, either \theeplayer wins or play reaches $\control'$ with the counters encoded on the stack incremented correctly.

\item
    Whenever we have
    $\tup{\control,
          \tup{\cha, \counter_\maxrank, \ldots, \counter_1},
          \pop{\opord},
          \tup{\control', \rank}}$
    followed by
    \[
        \tup{\tup{\control', \rank},
             \tup{\chb, \counter'_\maxrank, \ldots, \counter'_1},
             \rew{\tup{\chb, \counter''_\maxrank, \ldots, \counter''_1}},
             \control'}
    \]
    in $\gamebound$, then in $\gamepoly$ we have two cases.
    \begin{itemize}
    \item
        When $\opord > 1$ we have
        $\tup{\control, \cha, \pop{\opord}, \controlinc{\rank}{\cha}{\control'}}$
        after which, as argued in the previous section, either \theeplayer wins or play reaches $\control'$ with the counters encoded on the stack incremented correctly.

    \item
        When $\opord = 1$ we move to
        $\controlpop{\genplayer}{\rank}{\control'}$
        and can only perform a sequence of $\pop{1}$ operations to remove
        $\cha \stackw_\maxrank \ldots \stackw_1 \linkchar{\opord}$
        from the stack and then reaching
        $\controlincnoa{\rank}{\control'}$
        after which, as argued in the previous section, either \theeplayer wins or play reaches $\control'$ with the counters encoded on the stack incremented correctly.
        (Note, if $\pop{1}$ were not available in $\gamebound$, it would be unavailable when reaching $\linkchar{\opord}$, at which point the instigating player would lose the game.)
    \end{itemize}

\item
    Whenever we have
    \[
        \tup{\control,
             \tup{\cha, \counter_\maxrank, \ldots, \counter_1},
             \collapse,
             \tup{\control', \rank}}
    \]
    followed by
    \[
        \tup{\tup{\control', \rank},
             \tup{\chb, \counter'_\maxrank, \ldots, \counter'_1},
             \rew{\tup{\chb, \counter''_\maxrank, \ldots, \counter''_1}},
             \control'}
    \]
    in $\gamebound$, then in $\gamepoly$ we move to
    $\controlcollapse{\genplayer}{\rank}{\control'}$
    and can only perform a sequence of $\pop{1}$ operations to remove
    $\cha \stackw_\maxrank \ldots \stackw_1$
    and uncover
    $\linkchar{\opord}$
    to which we apply $\collapse$ and reach
    $\controlincnoa{\rank}{\control'}$.
    As argued in the previous section, either \theeplayer then wins or play reaches $\control'$ with the counters encoded on the stack incremented correctly.
    (As before, note that if the collapse is not available, this becomes apparent on reacing $\linkchar{\opord}$ where the instigating player loses the game.)
\end{itemize}
In each case we maintain the invariant.
That is, \theeplayer wins the safety game.

\subsection{%
    From
    $\ap{\polyred{\bound}}{\game}$
    to
    $\ap{\counterred{\bound}}{\game}$
}

Suppose $\strat$ is a winning strategy for \theeplayer in the game
$\gamepoly = \ap{\polyred{\bound}}{\game}$.
We define a winning strategy $\strat'$ for \theeplayer in
$\gamebound = \ap{\counterred{\bound}}{\game}$.
As before, the strategy is to essentially copy the winning play from
$\ap{\polyred{\bound}}{\game}$.

The initial configuration in $\gamepoly$ is
$\tup{\controlinit',
      \sbrac{\cdots\sbrac{
                \chainit'
            }{1}\cdots}{\cpdsord}}$
from which play is forced to
$\tup{\controlzero{1}{\chainit}{\control},
      \sbrac{\cdots\sbrac{
          \linkchar{1}
      }{1}\cdots}{\cpdsord}}$
and, by the previous section, either \theeplayer wins, or play reaches
$\tup{\control,
      \chainit
      \stackw_\maxrank \ldots \stackw_1
      \linkchar{1}}$
where each $\stackw_\rank$ encodes $0$.
The initial configuration in $\gamebound$ is the configuration
$\tup{\controlinit,
      \sbrac{\cdots\sbrac{
          \tup{\chainit, 0, \ldots, 0}
      }{1}\cdots}{\cpdsord}}$.
Thus, we have set up the invariant that when $\gamebound$ is in
$\tup{\control, \stackw}$
then $\gamepoly$ is in
$\ap{\polyred{\bound}}{\control, \stackw}$.

We show that for every sequence of moves played in $\gamepoly$ (between control states $\control$) there is a play in $\gamebound$ that re-establishes the invariant.
There are a number of cases.
Note, counters never overflow because $\strat$ is winning for \theeplayer.
In the following, suppose play has reached
$\tup{\control,
      \ccompose{\cha
                \stackw_\maxrank \ldots \stackw_1
                \annot{\linkchar{\opord}}{\tup{\opord, \idxi}}}
               {1}{\stackw}}$.
Let
$\stackw_\maxrank \ldots \stackw_1$
encode the counters
$\counter_1, \ldots, \counter_\maxrank$.
\begin{itemize}
\item
    When in $\gamepoly$ we have
    $\tup{\control,
          \cha,
          \push{\opord'},
          \controlincnoa{\rank}{\control'}}$
    then, as argued in the previous section, \theaplayer can force play to reach
    $\control'$
    after correctly incrementing the counters.
    That is, the topmost counter encodings become
    $\stackw'_\maxrank \ldots \stackw'_1$
    encoding the counters
    $\counter'_1, \ldots, \counter'_\maxrank
        = \ap{\counterinc{\rank}}{\counter_1, \ldots, \counter_\maxrank}$.

    In $\gamebound$ we can simply perform the following two moves:
    \[
        \tup{\control,
             \tup{\cha, \counter_1, \ldots, \counter_\maxrank},
             \push{\opord'},
             \tup{\control', \rank}}
    \]
    and
    \[
        \tup{\tup{\control', \rank},
             \tup{\cha, \counter_1, \ldots, \counter_\maxrank},
             \rew{\tup{\cha, \counter'_1, \ldots, \counter'_\maxrank}},
             \control'} \ .
    \]
    This maintains the invariant as required.

\item
    If we have in $\gamepoly$ a move
    $\tup{\control,
          \cha,
          \cpush{\chb},
          \controlinc{\rank}{\cha}{\control'}}$
    (with
     $\rank = \ap{\rankf}{\control'}$)
    then, as argued in the previous section, \theaplayer can play to first reach
    $\controlincnoa{\rank}{\control'}$
    with a copy of the previous counters.
    That is play reaches
    $\tup{\controlincnoa{\rank}{\control'}, \stackw'}$
    where
    $\stackw' =$
    \[
         \ccompose{\chb
                   \stackw_\maxrank \ldots \stackw_1
                   \annot{\linkchar{\opord'}}{\idxi'}
                   \cha
                   \stackw_\maxrank \ldots \stackw_1
                   \annot{\linkchar{\opord}}{\idxi}}
                  {1}{\stackw} \ .
    \]
    Then, as also argued previously, \theaplayer can continue play to $\control'$ with the counters encoded on the stack incremented correctly.
    That is play reaches
    \[
        \tup{\control',
             \ccompose{\chb
                       \stackw'_\maxrank \ldots \stackw'_1
                       \annot{\linkchar{\opord'}}{\idxi'}
                       \cha
                       \stackw_\maxrank \ldots \stackw_1
                       \annot{\linkchar{\opord}}{\idxi}}
                      {1}{\stackw}}
    \]
    where
    $\ap{\counterinc{\rank}}{\counter_1, \ldots, \counter_\maxrank}
        = \tup{\counter'_1, \ldots, \counter'_\maxrank}$
    and
    $\stackw'_\maxrank \ldots \stackw'_1$
    encodes
    $\counter'_1, \ldots, \counter'_\maxrank$.

    In this case, we can play
    \[
        \tup{\control,
             \tup{\cha, \counter_\maxrank, \ldots, \counter_1},
             \cpush{\tup{\chb, \counter_\maxrank, \ldots, \counter_1}}{\opord'},
             \tup{\control', \rank}}
    \]
    followed by
    \[
        \tup{\tup{\control', \rank},
             \tup{\chb, \counter_\maxrank, \ldots, \counter_1},
             \rew{\tup{\cha, \counter'_\maxrank, \ldots, \counter'_1}},
             \control'}
    \]
    in $\gamebound$.

\item
    When in $\gamepoly$ we have
    $\tup{\control,
          \cha,
          \rew{\chb},
          \controlincnoa{\rank}{\control'}}$
    then, as argued in the previous section, \theaplayer can force play to reach
    $\control'$
    after correctly incrementing the counters.
    That is, the topmost counter encodings become
    $\stackw'_\maxrank \ldots \stackw'_1$
    encoding the counters
    $\counter'_1, \ldots, \counter'_\maxrank
        = \ap{\counterinc{\rank}}{\counter_1, \ldots, \counter_\maxrank}$.

    In $\gamebound$ we can simply perform the following two moves:
    \[
        \tup{\control,
             \tup{\cha, \counter_1, \ldots, \counter_\maxrank},
             \rew{\tup{\chb, \counter_1, \ldots, \counter_\maxrank}},
             \tup{\control', \rank}}
    \]
    and
    \[
        \tup{\tup{\control', \rank},
             \tup{\cha, \counter_1, \ldots, \counter_\maxrank},
             \rew{\tup{\cha, \counter'_1, \ldots, \counter'_\maxrank}},
             \control'} \ .
    \]
    This maintains the invariant as required.

\item
    Whenever we play in $\gamepoly$ a move
    \[
        \tup{\control, \cha, \pop{\opord'}, \controlinc{\rank}{\cha}{\control'}}
    \]
    with $\opord' > 1$ then, as argued in the previous section, \theaplayer can then force play to reach $\control'$ with the counters encoded on the stack incremented correctly.

    In $\gamebound$ we can play the following two moves:
    \[
        \tup{\control,
             \tup{\cha, \counter_\maxrank, \ldots, \counter_1},
             \pop{\opord'},
             \tup{\control', \rank}}
    \]
    and
    \[
        \tup{\tup{\control', \rank},
             \tup{\chb, \counter'_\maxrank, \ldots, \counter'_1},
             \rew{\tup{\chb, \counter''_\maxrank, \ldots, \counter''_1}},
             \control'} \ .
    \]
    The counters
    $\counter'_1, \ldots, \counter'_\maxrank$
    are the counters at the top of the stack after the $\pop{\opord}$ and
    $\counter''_1, \ldots, \counter''_\maxrank$
    are the same counters after the appropriate increment.

\item
    Whenever we play in $\gamepoly$ a move
    \[
        \tup{\control,
             \cha,
             \pop{1},
             \controlpop{\genplayer}{\rank}{\control'}}
    \]
    then the only possible continuation is to completely remove
    $\cha \stackw_\maxrank \ldots \stackw_1 \linkchar{\opord}$
    from the stack.
    Then play moves to
    $\controlincnoa{\rank}{\control'}$
    and as argued in the previous section, \theaplayer can then force play to reach $\control'$ with the counters encoded on the stack incremented correctly.

    In $\gamebound$ we can play the following two moves:
    \[
        \tup{\control,
             \tup{\cha, \counter_\maxrank, \ldots, \counter_1},
             \pop{1},
             \tup{\control', \rank}}
    \]
    and
    \[
        \tup{\tup{\control', \rank},
             \tup{\chb, \counter'_\maxrank, \ldots, \counter'_1},
             \rew{\tup{\chb, \counter''_\maxrank, \ldots, \counter''_1}},
             \control'} \ .
    \]
    Note, the $\pop{1}$ is available because the instigating player $\genplayer$ would lose the game on reaching $\linkchar{\opord}$ if not.
    The counters
    $\counter'_1, \ldots, \counter'_\maxrank$
    are the counters at the top of the stack after the $\pop{1}$ and
    $\counter''_1, \ldots, \counter''_\maxrank$
    are the same counters after the appropriate increment.

\item
    Whenever we play in $\gamepoly$ a move
    \[
        \tup{\control,
             \cha,
             \collapse,
             \controlcollapse{\genplayer}{\rank}{\control'}}
    \]
    then the only possible continuation is to completely remove
    $\cha \stackw_\maxrank \ldots \stackw_1$
    from the stack, exposing
    $\annot{\linkchar{\opord}}{\idxi}$.
    Then $\collapse$ has to be performed and play moves to
    $\controlincnoa{\rank}{\control'}$
    and as argued in the previous section, \theaplayer can then force play to reach $\control'$ with the counters encoded on the stack incremented correctly.

    In $\gamebound$ we can play the following two moves:
    \[
        \tup{\control,
             \tup{\cha, \counter_\maxrank, \ldots, \counter_1},
             \collapse,
             \tup{\control', \rank}}
    \]
    and
    \[
        \tup{\tup{\control', \rank},
             \tup{\chb, \counter'_\maxrank, \ldots, \counter'_1},
             \rew{\tup{\chb, \counter''_\maxrank, \ldots, \counter''_1}},
             \control'} \ .
    \]
    The counters
    $\counter'_1, \ldots, \counter'_\maxrank$
    are the counters at the top of the stack after the $\pop{\opord}$ and
    $\counter''_1, \ldots, \counter''_\maxrank$
    are the same counters after the appropriate increment.
    The availability of the collapse move is guaranteed in the same way as for $\pop{1}$.
\end{itemize}
In each case we maintain the invariant.
That is, \theeplayer wins the safety game.

\end{document}